\documentclass[journal, 10pt]{IEEEtran}
\usepackage{amsthm,amsmath}
\usepackage{amsfonts}
\usepackage{amssymb}
\usepackage{mathrsfs}
\usepackage{mathtools}
\usepackage{float}
\usepackage[pdftex]{graphicx}
\usepackage{amsmath}
\usepackage{color}
\usepackage{multirow,multicol}
\usepackage{graphicx}
\usepackage{times}
\usepackage{textcomp}
\usepackage{verbatim}
\usepackage[table]{xcolor}
\usepackage{balance}
\usepackage{lipsum}
\usepackage[inline]{enumitem}
\usepackage[caption=false,font=footnotesize]{subfig}
\usepackage{cite}
\usepackage{algpseudocode,algorithm}
\usepackage{hyperref}
\usepackage{tcolorbox}
\usepackage{setspace,epstopdf}
\usepackage[table]{xcolor}
\definecolor{color1}{RGB}{199,209,232}
\definecolor{color2}{RGB}{230,231,233}

\DeclareMathOperator*{\argmax}{argmax} 
\DeclareMathOperator*{\maximize}{maximize} 
\DeclareMathOperator*{\minimize}{minimize} 
\DeclareMathOperator*{\subjectto}{subject\hspace{3pt} to \hspace{3pt}} 
\newtheorem{theorem}{Theorem}

\newtheorem{lemma}[theorem]{Lemma}

\begin{document}
	
	\title{Boosting Spectral Efficiency via Spatial Path \\ Index Modulation in RIS-Aided  mMIMO}
	
	
	\author{\IEEEauthorblockN{Ahmet M. Elbir, \textit{Senior Member, IEEE}, Abdulkadir Celik, \textit{Senior Member, IEEE}, \\ Asmaa Abdallah,  \textit{Member, IEEE}, and Ahmed M. Eltawil, \textit{Senior Member, IEEE}  }
		
		\thanks{The conference precursor of this work has been presented in IEEE PIMRC 2025  Conference \cite{spim_ris_conf1}.}
		
		\thanks{The work of A. M. Elbir is supported by TUBITAK (Turkish National Research Council) with the Grant No. 125E365.}
		
		\thanks{A. M. Elbir is with the Dept. of Electrical and Electronics Engineering at Istinye University, 34396 Istanbul, Turkey; and the Computer, Electrical, Mathematical Sciences \& Engineering Division, King Abdullah University of Science and Technology, Thuwal 23955, Saudi Arabia (e-mail: ahmetmelbir@gmail.com).}
		
		\thanks{A. Celik is with the School of Electronics and Computer Science, University of Southampton, Southampton SO17 1BJ, U.K. (e-mail: a.celik@soton.ac.uk).}
		
		\thanks{A. Abdallah and A. M. Eltawil are with the Computer, Electrical, Mathematical Sciences \& Engineering Division, King Abdullah University of Science and Technology, Thuwal 23955, Saudi Arabia (e-mail: asmaa.abdallah@kaust.edu.sa, ahmed.eltawil@kaust.edu.sa). } 
	}
	\maketitle
	
	\begin{abstract}
		Next generation wireless networks focus on improving spectral efficiency (SE) while reducing power consumption and hardware cost.	Reconfigurable intelligent surfaces (RISs) offer a viable solution to meet these requirements.	In order to enhance the SE, index modulation (IM) has been regarded as one of the enabling technologies via the transmission of additional information bits over the transmission media such as subcarriers, antennas and spatial paths. In this work, we explore the usage of spatial paths and introduce spatial path IM (SPIM) for RIS-aided massive multiple-input multiple-output (mMIMO) systems. Thus, the proposed framework  improves the network efficiency and the coverage with the use of RIS while SPIM provides SE improvement. In order to perform SPIM, we exploit the spatial diversity of the millimeter wave channel and assign the index bits to the spatial patterns of the channel between the base station and the users through RIS. We introduce a low complexity approach for the design of hybrid beamformers, which are constructed by the steering vectors corresponding to the selected spatial path indices for SPIM-mMIMO. Furthermore, we conduct a theoretical analysis on the SE of the proposed SPIM approach, and derive the SE relationship between the SPIM-based hybrid beamforming and fully digital (FD) beamforming. Via numerical simulations, we validate our theoretical results and show that the proposed SPIM approach presents an improved SE performance, even higher than that of the use of FD beamformers while using a few RF chains.
	\end{abstract}

	%
	\section{Introduction}
	\label{sec:Introduciton}
	\IEEEPARstart{N}{ext} generation wireless technologies focus on improving the energy and spectral efficiency (EE/SE) in response to the demand for new services, massive number of users and high data rate. Index modulation (IM) is one of the promising techniques as it offers both EE and SE improvement in comparison with the conventional modulation schemes~\cite{indexMod_Survey_Mao2018Jul}. 
	
	In IM, in addition to the information bits transmitted via conventional amplitude/phase modulation (APM), the indices of the transmission media such as subcarriers~\cite{im_basar_OFDM_Basar2013Aug}, antennas~\cite{antenna_grouping_SM} and spatial paths~\cite{spim_BIM_TVT_Ding2018Mar} are employed to transmit more information bits, as illustrated in Fig.~\ref{fig_IM}. By exploiting the diversity in these transmission entities, the transmitter encodes the information bits (i.e., index bits) in the indices of each of these transmission media, yielding a SE enhancement proportional to the IM media diversity. On the other hand, the use of IM media also constitutes a trade-off; for example, IM requires some of the subcarriers should be not utilized during on/off index bits assignment, thereby leading to a partial access to the whole spectrum~\cite{im_spatial_modulation_Mesleh2008Jul,spim_CapacityAnalysis_Guo2020Mar}. Similarly, in antenna index modulation, the full  array is not utilized and it causes loss in the SE~\cite{spim_secureSM_SubarraySelection_Shu2020Nov}. Another IM scheme is over the spatial paths between the transmitter and the receiver, wherein a portion of the channel paths are selected for encoding the index bits for each spatial pattern thereby conveying additional information bits~\cite{spim_GBM_Gao2019Jul,spim_bounds_JSTSP_Wang2019May}. Unlike the IM over frequency and antenna indices, SPIM is shown to be more robust against the performance loss due to the sacrificing the unused transmission media (e.g., subcarrier or antenna element) during information embedding~\cite{spim_GBMM_Guo2019Jul,spim_bounds_JSTSP_Wang2019May,elbir_IM_ISAC_Elbir2024Nov}. 
	
	Furthermore, researchers have been also aiming to reduce the hardware cost and power consumption to improve the EE and SE of the overall communication systems. In this respect, reconfigurable intelligent surfaces (RISs) offer a viable solution as one of the enabling technologies in the next generation wireless networks for EE improvement and extended coverage \cite{irs_TWC}. An RIS comprises a two-dimensional (2D) reconfigurable electromagnetic (EM) metasurface, constructed from a large periodic array of subwavelength scattering elements, commonly referred to as meta-atoms. The reflection coefficient of each RIS element can be locally adjusted in terms of both phase and polarization, enabling real-time manipulation of the direction of incident EM waves and thereby facilitating adaptive and programmable functionalities. The phase shifts of RIS elements are controlled by external signals transmitted from the base station (BS) via a dedicated backhaul control link \cite{irs_PLS_Arzykulov2023Mar}. Consequently, the signal transmitted from the BS can be dynamically redirected toward intended users, improving signal strength and coverage. By leveraging RIS, it becomes possible to significantly enhance the energy of received signals at distant user terminals, thereby extending the effective coverage area of the BS. Thus, combining the aforementioned two innovative approaches -- RIS and IM -- is the main motivation of this work for SE improvement of the communication system. Fig.~\ref{fig_IM}(d) shows the illustration of the RIS-aided communication system with IM.
	
	
	\begin{figure*}[t]
		\centering
		{\includegraphics[draft=false,width=\textwidth]{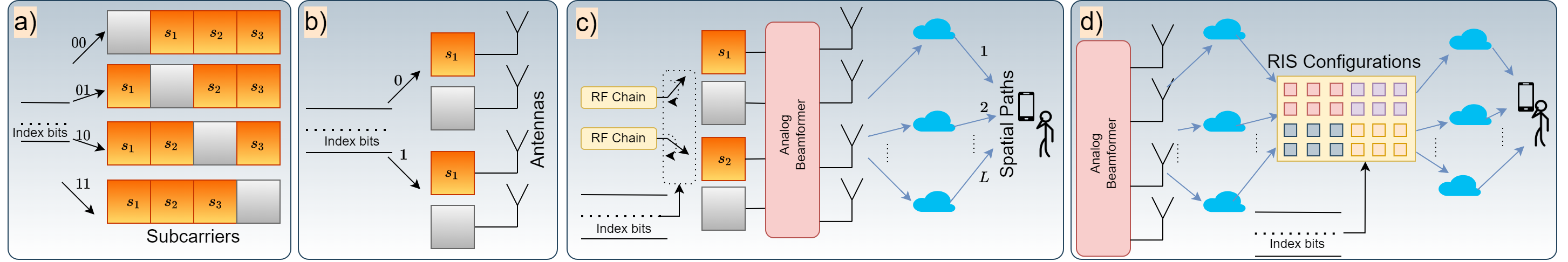} } 
		\caption{The illustration of the IM schemes over various transmission entities. In addition to the transmission of conventional APM bits, the index bits are assigned to the indices of  a) subcarriers, b) antennas, c) spatial paths and d) RIS configurations. 	}
		\label{fig_IM}
	\end{figure*}

	\subsection{Prior Work}
	
	In order to improve the SE in RIS-aided systems, various IM schemes have been employed in the existing works. In~\cite{irs_reflectionMod_Guo2020Jul}, reflection modulation technique has been introduced, wherein  the index bits are assigned to the reflection patterns corresponding to different RIS element configurations (see, e.g., Fig.~\ref{fig_IM}(d)). In order to achieve a similar purpose, a beam index modulation approach is presented by employing two closely located RISs in \cite{irs_BIM_Gopi2020Oct}, where the index bits are assigned to different RIS element configuration while assuming additive white Gaussian noise (AWGN) channel between the RISs. Another approach to diversify the RIS elements is to divide the RIS elements into multiple groups, to which the index bits are assigned for conveying additional information bits as in \cite{im_irs_grouping_Asmoro2022Sep}. To further improve the effective gain of the RIS while exploiting the RIS groups, the process of the assignment of the index bits is performed according to the channel gain of the received paths at the RIS in 	\cite{im_irs_offset_Zhang2024Mar}. As a result, not all of the RIS elements are selected, thereby generating an offset while avoiding the frequent controller operation between the BS and the RIS. Also in \cite{im_RIS_elements_Matemu2024Nov}, the RIS elements are divided into microstrips and the RIS elements are selected microstrip-wise for IM. In \cite{im_irs_superimposed_Yao2023Aug}, superimposed RIS phase modulation scheme is proposed, wherein tunable	phase offsets are superimposed onto predetermined RIS phases to embed the index bits.  In \cite{im_RIS_antennas_Basar2020Feb}, antenna indices are employed in RIS-aided scenario and RIS-space shift keying (RIS-SSK) and RIS-spatial modulation (RIS-SM) schemes are introduced. This approach is further generalized in \cite{   im_irs_groupoing_receive_antenna_Zhang2021Oct} by incorporating the grouping of RIS elements together with antenna indices.  To implement IM in frequency domain, the authors in \cite{im_irs_subcarrier_Hodge2020Sep} perform IM over the subcarrier indices in orthogonal frequency division multiplexing (OFDM) systems. This approach is further extended in \cite{im_irs_subcarrier2_Hodge2023Oct} to perform IM over multiple domains, i.e., subcarriers, time slots and antennas. This approach is also referred to as multi-dimensional IM~\cite{im_irs_multi_dim_Tusha2020Dec}. {Further exploitation of IM by utilizing the RIS also includes differential chaotic shift keying (DCSK) \cite{ref_R2_1_10844041,ref_R2_2_10086558}. }  {Also in \cite{ref_R1_4_9217944,ref_R1_3_10640072}, The RIS configuration is utilized and reflection pattern modulation techniques have been employed. }
	
	The aforementioned works mostly consider the RIS configurations for IM while the spatial diversity of the wireless channel is not exploited. This is especially useful in millimeter wave (mmWave) channels as they offer non-line-of-sight (nLoS) and multipath rich propagation environments in  massive multiple-input multiple-output (mMIMO) systems~\cite{valenzuela_Saleh2003Jan,rappaport2017Aug}. In spatial path IM (SPIM), the spatial paths between the BS and the user through the RIS can be exploited and the index bits can be assigned to the spatial patterns that involve the selected paths. Compared to other IM techniques, SPIM is shown to be more robust against the performance loss due to the sacrificing the unused transmission media during information embedding~\cite{spim_GBMM_Guo2019Jul,spim_bounds_JSTSP_Wang2019May,elbir_IM_ISAC_Elbir2024Nov}. In particular, authors in \cite{spim_GBMM_Guo2019Jul} have shown that the SPIM-aided communication with hybrid analog/digital beamforming can achieve higher SE than that of the use of fully digital (FD) beamformers. {The outperforming performance of hybrid beamforming with IM over the use of FD beamformers is attributed to delivering additional information bits via IM, which improves the overall communication performance in terms of SE.} Similar observations have also been achieved in \cite{spim_bounds_JSTSP_Wang2019May,elbir_IM_ISAC_Elbir2024Nov} and \cite{elbir2024Mar_SPIM_ISAC}. In particular, the SPIM-aided communication is shown to be  advantageous if the channel gains of the spatial paths used for SPIM are relatively close \cite{spim_bounds_JSTSP_Wang2019May}. This is intuitively useful as the SE of the communication system with/without SPIM directly depends on the channel gains. Furthermore, the design of hybrid beamformers for maximizing the SE, the analog beamformer components correspond to the path directions with strongest channel gains. As a result, SPIM techniques have the advantage of having full access of spectrum and antenna array while performing IM among the  spatial paths with comparable channel gains. The SPIM has been studied in  \cite{spim_BIM_TVT_Ding2018Mar} by designing the analog beamformers in accordance with the spatial patterns for a single-user with a single radio frequency (RF) chain scenario. In other words, only one of the spatial paths is selected to represent the index bits for each spatial pattern. Also in \cite{spim_bounds_JSTSP_Wang2019May} a single RF chain scenario is considered, and the theoretical analysis on the comparison of SPIM with conventional hybrid beamforming is presented. The same approach is extended in \cite{spim_GBMM_Guo2019Jul} for multiple RF chains, and the spatial patterns are derived from the FD beamformer. Instead of conventional antenna arrays, beamspace modulation was exploited  in\cite{spim_GBM_Gao2019Jul}  by employing lens arrays at both transmitter and receiver. The further extensions of the SPIM approach include by applying distributed machine learning \cite{spim_FL_Elbir2021Jun} and  integrated sensing and communication  \cite{elbir2024Mar_SPIM_ISAC}. 
	
	In summary, the aforementioned SPIM works~\cite{spim_BIM_TVT_Ding2018Mar,spim_bounds_JSTSP_Wang2019May,spim_FL_Elbir2021Jun,elbir2024Mar_SPIM_ISAC,spim_GBM_Gao2019Jul,spim_GBMM_Guo2019Jul,elbir_IM_ISAC_Elbir2024Nov} consider the mmWave mMIMO scenario without exploiting the advantage of RIS for EE and network coverage. Furthermore, most of these works assume single-user scenario and do not consider the multi-user SPIM-aided mMIMO, for which the derivation of the SE is rather challenging. While there exist IM techniques for RIS-aided mMIMO systems, SPIM is not considered. Instead, these approaches employ different transmission entities, i.e., the RIS element configurations \cite{irs_reflectionMod_Guo2020Jul,im_irs_grouping_Asmoro2022Sep,im_irs_offset_Zhang2024Mar,irs_BIM_Gopi2020Oct,im_irs_superimposed_Yao2023Aug,im_RIS_elements_Matemu2024Nov}, antennas \cite{im_RIS_antennas_Basar2020Feb,im_irs_groupoing_receive_antenna_Zhang2021Oct} or  subcarriers \cite{im_irs_subcarrier_Hodge2020Sep,im_irs_subcarrier2_Hodge2023Oct,im_irs_multi_dim_Tusha2020Dec} for IM.  Therefore, the motivation of this work is improve the SE of the RIS-aided mMIMO systems via  SPIM techniques.

\begin{figure*}[t]
	\centering
	{\includegraphics[draft=false,width=.8\textwidth]{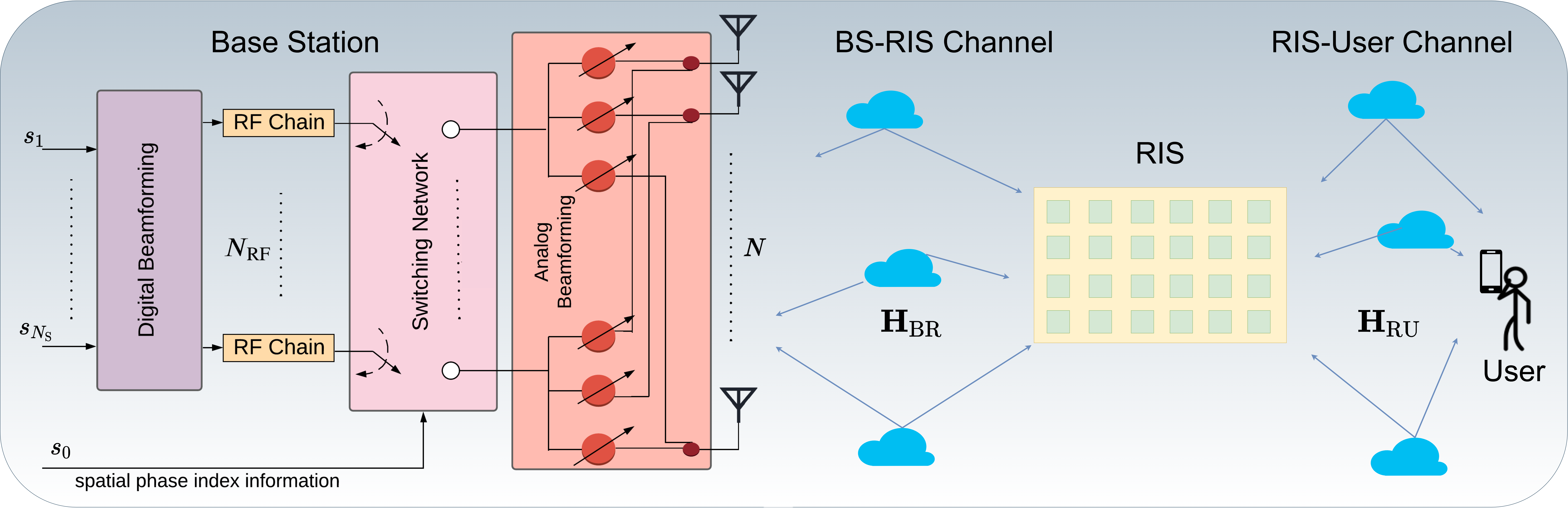} } 
	\caption{SPIM for RIS-aided mMIMO system, wherein the BS processes the incoming data-streams and employs spatial path index information $s_0$ in a switching network, which connects $N_\mathrm{RF}$ RF chains to $L$ taps on the analog beamformer to exploit $L_\mathrm{S}$ out of $L$ paths for communication.{In UL, the estimated spatial path directions at the user are fed-back to the BS which performs SPIM beamformer design. In DL, the user measures the path strengths to detect the spatial patterns.     }    }
	\label{fig_IRS_MIMO}
\end{figure*}

\subsection{Our Contributions}

In this paper, we introduce an SPIM-based hybrid beamforming approach for RIS-aided mMIMO systems, illustrated in Fig.~\ref{fig_IRS_MIMO}. 		In this regard, we first introduce the system design for single-user scenario. Then, the proposed method is extended for multi-user RIS-aided mMIMO systems.  In each scenario, the diversity of the spatial paths in the mmWave channel is exploited for the users. In particular, a portion of the spatial paths between the BS and users through RIS are selected for each spatial pattern and their indices are assigned to index bits.	Via numerical simulations and theoretical analysis, we have shown that the proposed SPIM approach based on hybrid beamforming achieve higher SE than that of FD beamforming thanks to delivering additional index bits.	Our main contributions are summarized as follows:
\begin{enumerate}
	\item Contrary to the existing works, we leverage SPIM for RIS-aided  mMIMO systems in order to obtain higher SE while extending the coverage of the wireless network. To maximize the SE of the overall system, the optimization problem is handled with two subproblems for designing the hybrid analog and digital beamformers at the BS, and finding the phase shift configuration of RIS elements. 
	\item In order to design the hybrid beamformers at the BS, the channels for BS-RIS and RIS-user links are estimated through  channel acquisition in the uplink. Then, the analog beamformer at the BS is designed such that the columns of analog beamformer are selected as the steering vectors corresponding to the channel path directions for the selected spatial pattern. Next, the RIS phase elements are estimated for the given SPIM scenario. In particular, the SE expression for FD beamforming is utilized to obtain a convex optimization problem, and the RIS phase elements are found. 
	\item We expand our SPIM approach for multi-user scenario, and derive the SE expression by taking into account the contribution of the index bits together with the conventional APM bits.
	\item Based on the proposed SPIM approach, we conduct theoretical performance analysis on the SE, and derive the SE relationship between the FD beamforming and the SPIM-aided hybrid beamforming. We further validate our theoretical analysis by numerical simulations showing that our SPIM approach achieves higher ($\sim 20\%$) SE than that of FD beamforming.
	
	
\end{enumerate}

\subsection{Notation} Throughout the paper,  $(\cdot)^\textsf{T}$ and $(\cdot)^{\textsf{H}}$ denote the transpose and conjugate transpose operations, respectively. For a matrix $\mathbf{A}$ and vector $\mathbf{a}$; $[\mathbf{A}]_{ij}$, $[\mathbf{A}]_k$  and $[\mathbf{a}]_l$ correspond to the $(i,j)$-th entry, $k$-th column and $l$-th entry, respectively. $\lfloor\cdot \rfloor$  represents the flooring  operation. We denote $|| \cdot||_2$ and $|| \cdot||_\mathcal{F}$ as the  $l_2$-norm and Frobenious norm, respectively.

%

%

\section{System Model }
Consider an RIS-assisted mmWave massive MIMO system, as shown in Fig.~\ref{fig_IRS_MIMO}, wherein the BS is equipped with a uniform linear array (ULA) of $N$ antennas with hybrid beamforming structure, which consists of $N_\mathrm{RF}$ RF chains. The BS communicates with the user employing $\bar{N}$ antennas through an RIS while the direct LoS link is blocked. The RIS has a uniform planar array (UPA) structure on $yz$-plane, and it is composed of $M$ passive reflecting surface elements to assist the communication. In the considered scenario, the BS communicates with the user via $N_\mathrm{S}$ data-streams. In particular, the BS sends $\mathbf{s} = [s_1,\cdots, s_{N_\mathrm{S}}]^\textsf{T}$ ($\mathbb{E}\{\mathbf{ss}^\textsf{H} \} = \mathbf{I}_{N_\mathrm{S}}$), via conventional APM schemes.  In addition to $\mathbf{s}$, the spatial path index information, which is represented by $s_0$, is fed to a switching network as shown in Fig~\ref{fig_IRS_MIMO}, to assign the outputs of $N_\mathrm{RF}$ RF chains to $L$ taps of the analog beamformer. Here, $L$ is defined as the total number of spatial paths incoming to the BS~\cite{elbir_IM_ISAC_Elbir2024Nov}. Here, $L = \mathrm{min}\{L_\mathrm{BR},L_\mathrm{RU}\}$, where $L_\mathrm{BR}$ and $L_\mathrm{RU}$ denote the number paths for the BS-RIS and RIS-user links. In order to perform SPIM, we assume that $L_\mathrm{S}$ out of $L$ paths are selected for communication. As a result, we have $N_\mathrm{RF} = L_\mathrm{S}$. Consequently, performing SPIM yields $\footnotesize \left(\begin{array}{c}
	L \\
	L_\mathrm{S}
\end{array}\right)$ choices of connection to exploit the diversity of the spatial paths. Thus, the BS can process at most $L_\mathrm{S}$ inputs, each of which is connected to $N> L$ antennas via phase shifters. Next, we define the total number of spatial patterns as
$\footnotesize S = 2^{\left\lfloor \log_2 \left(\begin{array}{c}
		L \\
		L_\mathrm{S}
	\end{array}\right)\right\rfloor},$
which implies that the spatial domain can carry at most ${\footnotesize \left\lfloor \log_2 \left(\begin{array}{c}
		L \\
		L_\mathrm{S}
	\end{array}\right)\right\rfloor}$ bits of information. {We also note that this switching operation at the precoders and at the RIS phase switches needs to be conducted in accordance with the symbol duration, for which low-cost switching components with the speed of nanoseconds are available~\cite{spim_GBMM_Guo2019Jul,heath2016overview}. Therefore, accurate switching performance of the beamformer, especially analog beamformer with phase-shifters, is of great importance. In~\cite{switch_phase_shifter_2_10354030}, the high speed switching performance is achieved for the design of phase shifters  with 2.1 dB insertion loss on 26 GHz band. Also in~\cite{spim_GBMM_Guo2019Jul}, low-cost phase shifter designs are presented with the switching speed in the range of tens of nanoseconds.  Thanks to these high speed switching capabilities, the implementation of SPIM is considered to be feasible.} As a special case, when $L_\mathrm{S} = L$, the presented SPIM configuration reduces to a conventional mmWave system as there is only a single choice of communication.

	
	For the $i$-th spatial pattern ($i = 1,\cdots, S$), the BS processes the data-streams $\mathbf{s}$ via baseband beamformer $\mathbf{B}_i\in \mathbb{C}^{N_\mathrm{RF}\times N_\mathrm{S}}$ and analog beamformer $\mathbf{A}_i\in \mathbb{C}^{N\times N_\mathrm{RF}}$, where $N_\mathrm{S}\leq N_\mathrm{RF}\leq N$. The analog beamformer has constant-modulus constraint as $|[\mathbf{A}_i]|_{n,r} = 1/\sqrt{N}$ for $n = 1,\cdots, N$ and $r = 1,\cdots, N_\mathrm{RF}$. Furthermore, we have the power constraint as $\|\mathbf{A}_i \mathbf{B}_i\|_\mathcal{F}^2 = N_\mathrm{S}$. As a result, the $N\times 1$ transmitted signal becomes
	\begin{align}
		\mathbf{x}_i = \mathbf{A}_i\mathbf{B}_i\mathbf{s},
	\end{align}
	which propagates through the mmWave channel via BS-RIS and RIS-UE links.	In mmWave transmission,  the wireless channel exhibits a sparse multipath structure, and it is usually characterized as the contribution of $L$ path components by the Saleh-Valenzuela (SV) channel model~\cite{valenzuela_Saleh2003Jan}. Thus, we define the downlink channels for the BS-RIS and RIS-user links as $\mathbf{H}_\mathrm{BR}\in\mathbb{C}^{M\times N}$ and $\mathbf{H}_\mathrm{RU}\in \mathbb{C}^{\bar{N}\times M}$, respectively. 
	In RIS-assisted scenario, the RIS elements in the RIS combines the incoming signals and re-scatter them with adjustable phase shifts toward the user~\cite{ris_Design_SU2_Hong2022Mar,ris_Design_SU_Zhu}. Define $\psi_m$ as the phase shift introduced by the $m$-th RIS element, where $\psi_m \in \{\frac{2\pi b}{2^\Delta} | b = 0,\cdots, 2^\Delta -1  \}$ is the set of discrete RIS phases. Then, we define $\boldsymbol{\Psi} = \mathrm{diag}\{\boldsymbol{\psi}\}\in \mathbb{C}^{M\times M}$  as the phase shift matrix
	where $\boldsymbol{\psi} = \left[e^{\mathrm{j}\psi_1},\cdots, e^{\mathrm{j}\psi_M}\right]^\textsf{T}$. Thus, the ${\bar{N} \times N}$ cascaded downlink channel matrix $\mathbf{H}$ from BS to the user through RIS is defined as
	\begin{align}
		\mathbf{H}=  \mathbf{H}_\mathrm{RU}\boldsymbol{\Psi}\mathbf{H}_\mathrm{BR}.
		\label{channelCascaded}
	\end{align}
	Finally,		the $\bar{N}\times 1$ received signal for the $i$-th spatial pattern at the user becomes		 
	\begin{align}
		{\mathbf{y}}_i =  \mathbf{H}\mathbf{A}_i\mathbf{B}_i\mathbf{s} + \mathbf{n}, \label{receivedSignal1}
	\end{align}
	where  {$\mathbf{n}\in \mathbb{C}^{\bar{N}}$ denotes the zero-mean AWGN vector, i.e.,  $\mathbf{n}\sim \mathcal{CN}(\textbf{0},\sigma_n^2\textbf{I}_{\bar{N}})$ with variance $\sigma_n^2$}.

	%

	\section{SE Analysis \& Problem Formulation}
	We aim to design the hybrid beamformers $\mathbf{A}_i$ and $\mathbf{B}_i$ so that we maximize the SE, which is characterized by the mutual information  (MI) of the communication links between the BS and the user~\cite{spim_AsymptoticMI_He2017Nov,spim_GBMM_Guo2019Jul}. In the following, we first discuss the SE of  the SPIM-aided mMIMO system, i.e., $I_\mathrm{SPIM}$, and the conventional mmWave mMIMO system, i.e., $I_\mathrm{MIMO}$. Then, we introduce the hybrid beamforming design problem.
	
	\subsection{SE for SPIM mMIMO}
	Define  $\mathbf{F}_i \in \mathbb{C}^{N\times N_\mathrm{RF}}$ and $\mathbf{x}_i\in \mathbb{C}^{N}$ as the hybrid beamformer for the $i$-th spatial pattern, i.e., $\mathbf{F}_i = \mathbf{A}_i\mathbf{B}_i$, and the transmit signal, i.e., $\mathbf{x}_i = \mathbf{F}_i\mathbf{s}$, respectively. Then, the SE of the SPIM-aided system is given by
	\begin{align}
		I_\mathrm{SPIM} = I({\mathbf{y}}_{i};\mathbf{x}_i, \mathbf{F}_{i}), \label{SE_10}
	\end{align}
	where $I({\mathbf{y}}_{i};\mathbf{x}_i|\mathbf{F}_{i})$ represents the MI between the $i$-th spatial pattern and the received signal ${\mathbf{y}}_{i}$, and it can be expressed as
	\begin{align}
		I({\mathbf{y}}_{i};\mathbf{x}_i, \mathbf{F}_{i}) = I({\mathbf{y}}_{i};\mathbf{x}_i|\mathbf{F}_{i}) + I({\mathbf{y}}_{i};\mathbf{F}_{i}), \label{SE_11}
	\end{align}
	where $I({\mathbf{y}}_{i};\mathbf{x}_i|\mathbf{F}_{i})$ and $I({\mathbf{y}}_{i};\mathbf{F}_{i})$ correspond to the information bits conveyed via conventional APM and SPIM, respectively. In particular, $I({\mathbf{y}}_{i};\mathbf{x}_i|\mathbf{F}_{i})$ corresponds to the conventional symbol transmission, and it can be quantified by using Shannon's continuous-input continuous-output memoryless channel's (CCMC) capacity\footnote{While CCMC is considered in this work, the discrete-input continuous-output memoryless channel (DCMC) capacity may also be taken into account as DCMC utilizes the choice of specific digital modulation schemes~\cite{ref_R1_1_1608632,ref_R1_2_10250854}.    } as~\cite{capacity_Telatar1999Nov}
	\begin{align}
		I({\mathbf{y}}_{i};\mathbf{x}_i|\mathbf{F}_{i}) = \frac{1}{S}\sum_{i = 1}^{S}\log_2 \left(\left| \frac{\mathbf{M}_i}{\sigma_n^2}\right|\right), \label{SE_12}
	\end{align}
	where $	\mathbf{M}_i \in \mathbb{C}^{\bar{N}\times \bar{N}}$ is the covariance matrix of the received signal in (\ref{receivedSignal1}) as 
	\begin{align}
		\mathbf{M}_i = \sigma_n^2\mathbf{I}_{\bar{N}} + \frac{1}{ N_\mathrm{S}} \mathbf{H}\mathbf{A}_i\mathbf{B}_i \mathbf{B}_i^\textsf{H}\mathbf{A}_i^\textsf{H}\mathbf{H}^\textsf{H}. \label{m_i}
	\end{align}
	Next, we consider the second term in (\ref{SE_11}), for which there is no closed-form expression for $I({\mathbf{y}}_{i};\mathbf{F}_{i})$. However, it is lower-bounded by $I_\mathrm{LB}({\mathbf{y}}_{i};\mathbf{F}_{i})$, which is defined as
	\begin{align}
		I_\mathrm{LB}({\mathbf{y}}_{i};\mathbf{F}_{i}) = &\log_2 (S) - \bar{N}\log_2 (e) \nonumber \\
		& - \frac{1}{S}\sum_{i=1}^{S} \log_2 \left(\sum_{j = 1}^{S} \frac{|\mathbf{M}_i|}{|\mathbf{M}_i + \mathbf{M}_j|} \right), \label{SE_13}
	\end{align}
	which is shown to be a tight approximation of $I({\mathbf{y}}_{i};\mathbf{F}_{i})$~\cite{spim_SE_He2017May}. {Following (\ref{SE_11}), in order to obtain the SE expression for SPIM, we combine (\ref{SE_13}) and  (\ref{SE_12}) as  
			\begin{align}
			&I_\mathrm{SPIM} = I({\mathbf{y}}_{i};\mathbf{x}_i|\mathbf{F}_{i}) + 	I_\mathrm{LB}({\mathbf{y}}_{i};\mathbf{F}_{i}).
					\end{align}
					Thus, we first rewrite (\ref{SE_12}) as
						\begin{align}
						\frac{1}{S}\sum_{i = 1}^{S}\log_2 \left(\left| \frac{\mathbf{M}_i}{\sigma_n^2}\right|\right) =  \log_2 \left(\frac{1}{\sigma_n^{^{2\bar{N}} }}\right)  + \frac{1}{S} \sum_{i = 1}^S \log_2 | \mathbf{M}_i|, \label{SE_22}
					\end{align}
					and the last term in the right hand side of (\ref{SE_13}) as 
						\begin{align}
					&	\frac{1}{S}\sum_{i=1}^{S} \log_2 \left(\sum_{j = 1}^{S} \frac{|\mathbf{M}_i|}{|\mathbf{M}_i + \mathbf{M}_j|} \right) \nonumber  \\
						& = \frac{1}{S} \sum_{i = 1}^S \log_2 | \mathbf{M}_i|  + \frac{1}{S}\sum_{i=1}^{S} \log_2 \left(\sum_{j = 1}^{S} \frac{|\mathbf{M}_i|}{|\mathbf{M}_i + \mathbf{M}_j|} \right). \label{SE_23}
					\end{align}
					Then, using (\ref{SE_22}) and (\ref{SE_23}), we combine (\ref{SE_13}) and  (\ref{SE_12}) as 
			\begin{align}
					I_\mathrm{SPIM} =& \log_2(S) +  \log_2 \left(\frac{1}{\sigma_n^{^{2\bar{N}} }}\right) + \log_2 \left(2 ^{- \bar{N}}\right) \nonumber \\
					& -\frac{1}{S}\sum_{i=1}^{S} \log_2 \left(\sum_{j = 1}^{S} \frac{|\mathbf{M}_i|}{|\mathbf{M}_i + \mathbf{M}_j|} \right),
		\end{align}
		wherein the constant gap term, i.e., $\bar{N}(1-\log_2 (e))$ is deleted to compensate the asymptotic approximation error~\cite{spim_AsymptoticMI_He2017Nov}. Finally, the  SE of the SPIM-aided system is given by
		\begin{align}
			I_\mathrm{SPIM} = \log_2\left(\frac{S}{(2\sigma_n^2)^{\bar{N}}}\right)   -  \frac{1}{S}  \sum_{i = 1}^{S}  \log_2  \left(\sum_{j = 1}^{S}\frac{1}{ | \mathbf{M}_i + \mathbf{M}_j |  }  \right). \label{SE_14}
		\end{align}

}

	\subsection{SE for mMIMO}
	The SE of the conventional mmWave mMIMO systems performing only APM is achieved by steering the output of the RF chains to the strongest paths. Therefore, according to Shannon's formula, 	the SE of the mMIMO system is given by
	\begin{align}
		I_\mathrm{MIMO} = \log_2 \left(\frac{1}{\sigma_n^2}\left| \mathbf{M}_1 \right|\right). \label{se_mmWave}
	\end{align}
	Since there is no SPIM, there is only a single choice of transmission, i.e., $S = 1$, for which we have $I_\mathrm{SPIM} = I_\mathrm{MIMO}$, as shown in the following lemma.
	\begin{lemma}
		\label{lemma1}
		The SE of the SPIM-aided system is equivalent to the SE of  conventional mMIMO system if $S = 1$.
	\end{lemma}

	\begin{proof}
		The proof is provided in Appendix~\ref{proofL1}.
	\end{proof}
	
	In order to get more insight, we further define the FD beamformer and show the relationship between the SE of SPIM-aided and FD mMIMO systems. We first define the SVD of the cascaded channel of the RIS-mMIMO system as $			\mathbf{H} = \mathbf{U}\boldsymbol{\Sigma}\mathbf{V}^\textsf{H},$
	where $\mathbf{U}\in \mathbb{C}^{\bar{N}\times \bar{N}}$ and $\mathbf{V}\in \mathbb{C}^{N\times N}$ are left and right singular value matrices, respectively~\cite{heath2016overview}. $\boldsymbol{\Sigma} = \left[\begin{array}{cc}
		\boldsymbol{\Sigma}_1 & \mathbf{0} \\
		\mathbf{0} & \boldsymbol{\Sigma}_2
	\end{array}\right]  \in \mathbb{C}^{\bar{N}\times N}$ is composed of the singular values in descending order, for which $\boldsymbol{\Sigma}_1\in \mathbb{C}^{N_\mathrm{S}\times N_\mathrm{S}}$ includes the $N_\mathrm{S}$ largest singular values. Now, the FD beamformer $\mathbf{F}\in \mathbb{C}^{N\times N_\mathrm{S}}$ can be defined as columns of the singular matrix corresponding to the $N_\mathrm{S}$ largest singular values, i.e., $\mathbf{F} = \mathbf{V}_1$. Then, the SE for the FD beamforming is given by
	\begin{align}
		I_\mathrm{FD} = \log_2\left( \left| \mathbf{I}_{N_\mathrm{S}} + \frac{1}{\sigma_n^2 N_\mathrm{S}} \boldsymbol{\Sigma}_1^2  \right|  \right). \label{se_fd}
	\end{align}
	
	In the following theorem, we show the relationship between the SE of SPIM-aided and FD mMIMO systems. We later validate the theoretical results with numerical experiments and show that the proposed SPIM approach achieves higher SE than that of the FD mMIMO beamforming.
	
	\begin{theorem}
		\label{thereom1}
		Let $	I_\mathrm{SPIM}$ and $I_\mathrm{FD}$ be the SE of the mMIMO communication systems with SPIM hybrid beamforming and FD beamforming. Then, they have the following relationship, i.e.,
		\begin{align}
			I_\mathrm{SPIM} -I_\mathrm{FD}\geq   \log_2\left(\frac{S}{4}\right) -N_\mathrm{S} -\tau, \label{thereom1eq}
		\end{align}
		where $\tau = \frac{1}{S}\log_2\left(\prod_{i = 1}^{S}\sum_{j = 1}^S 2^{-(  u_i+u_j)}  \right)$ for $u_z = ||\mathbf{V}_1^\textsf{H}\mathbf{A}_z \mathbf{B}_z ||_\mathcal{F}^2 $ for $z \in \{i,j\}$.
	\end{theorem}
	
	\begin{proof}
		The proof is provided in Appendix~\ref{proofT1}.
	\end{proof}
	
	\subsection{Problem Formulation}
	
	The hybrid beamformer design problem can be formed by maximizing the SE of the overall SPIM-aided RIS-mMIMO system, i.e.,
	\begin{subequations}
		\begin{align}		\maximize_{\mathbf{A}_i,\mathbf{B}_i,\boldsymbol{\Psi} } 	& \  I_\mathrm{SPIM} \nonumber \\
			\subjectto \ 	&\mathbf{A}_i \in \mathcal{A}, \mathbf{B}_i \in \mathcal{B}, \left|\left[\mathbf{A}_i\right]_{n,r} \right| = 1/\sqrt{N} \label{opt1_1} \\
			&\|\mathbf{A}_i \mathbf{B}_i  \|_\mathcal{F}^2 = N_\mathrm{S}, |[\boldsymbol{\psi}]_m| =1, \; \forall m \label{opt1_3} ,
		\end{align}
		\label{opt1}
	\end{subequations}
	where $\mathcal{A} = \{\mathbf{A}_1,\cdots, \mathbf{A}_S  \}$ and $\mathcal{B} = \{\mathbf{B}_1,\cdots, \mathbf{B}_S  \}$ denote the set of possible analog beamformers for SPIM and the corresponding baseband beamformers, respectively. The optimization problem in \eqref{opt1} is non-convex due to the constant-modulus constraint of the analog beamformers in \eqref{opt1_1} and the RIS phase terms in (\ref{opt1_3}), and nonlinear due to multiplications of multiple unknown variables $\mathbf{A}_i,\mathbf{B}_i$ and $\boldsymbol{\Psi}$. In order to solve \eqref{opt1} effectively, we present our solution in the following section.


	\section{SPIM for RIS-mMIMO}
	In case of no RIS, the maximization of $I_\mathrm{SPIM}$ in (\ref{opt1}) is equivalent to maximizing $I_\mathrm{MIMO}$ in (\ref{se_mmWave}) for all spatial patterns as the SE of SPIM corresponds to the contribution of all analog beamformer candidates in $\mathcal{A}$~\cite{spim_bounds_JSTSP_Wang2019May,spim_onGSM_He2017Sep}. This is done by first obtaining the set of possible analog beamformers, i.e., $\mathcal{A}$, and then, $I_\mathrm{SPIM}$ is computed by employing $\mathbf{M}_i$ for each $\mathbf{A}_i$, $i = 1,\cdots, S$. {In order to provide a solution for the RIS-assisted scenario, the joint problem in (\ref{opt1}) is  decomposed into two separate problems and solved alternatingly for the design of hybrid beamformers and the RIS phases. While the optimality of decoupling the problem is not guaranteed, this alternating approach has been shown to achieve convergence with sufficient accuracy~\cite{irs_CE_practical_MU_Di2020Feb,ris_Design_SU2_Hong2022Mar,ris_Design_SU_Zhu}.} In particular, we first design the hybrid beamformer $\mathbf{F}_i = \mathbf{A}_i \mathbf{B}_i$. Specifically, we design the columns of the analog beamformer $\mathbf{A}_i$ from the set of steering vectors corresponding to the path directions as this approach provides ``close-to-optimum" solution, thereby making the the eigenvalues of the matrix $\mathbf{I}_{N_\mathrm{S}} - \mathbf{V}_1^\textsf{H}\mathbf{F}_i\mathbf{F}_i^\textsf{H}$ small, i.e., $\mathbf{F}^\textsf{H}\mathbf{F}_i\mathbf{F}_i^\textsf{H}\mathbf{F} \approx \mathbf{I}_{N_\mathrm{S}}$~\cite{mimoRHeath,spim_bounds_JSTSP_Wang2019May}.  Then, in the second problem, we consider the optimization of the RIS phase matrix $\boldsymbol{\Psi}$.


	
	\subsection{Channel Acquisition and Beamformer Design}
	The maximization of the cost function in (\ref{opt1}) is achieved by designing the analog beamformer $\mathbf{A}_i$ from the steering vectors corresponding to the directions of the paths~\cite{spim_bounds_JSTSP_Wang2019May,spim_AsymptoticMI_He2017Nov}. Therefore, we design the analog the beamformer as the collection of the steering vectors corresponding to the incoming signal paths, which can be obtained from the uplink channel acquisition stage. Define the uplink channels ${\mathbf{G}}_\mathrm{BR}  = \mathbf{H}_\mathrm{BR}^\textsf{T} \in \mathbb{C}^{N\times M}$ and  ${\mathbf{G}}_\mathrm{RU}  = \mathbf{H}_\mathrm{RU}^\textsf{T} \in \mathbb{C}^{M\times \bar{N}}$ as
	\begin{align}
		{\mathbf{G}}_\mathrm{BR} &= \sqrt{\frac{N M}{L_\mathrm{BR}}} \sum_{\ell = 1}^{L_\mathrm{BR}} \alpha_\ell \mathbf{a}_{\mathrm{BS}}(\vartheta_\ell^r)\mathbf{a}_{\mathrm{RIS}}^\textsf{H}(\theta_\ell^t, \phi_\ell^t) ,  \\
		{\mathbf{G}}_\mathrm{RU} &= \sqrt{\frac{M\bar{N} }{L_\mathrm{RU}}} \sum_{\ell = 1}^{L_\mathrm{RU}} \beta_\ell \mathbf{a}_{\mathrm{RIS}}(\theta_\ell^r, \phi_\ell^r)\mathbf{a}_{\mathrm{UE}}^\textsf{H}(\vartheta_\ell^t) ,
		\label{channelModels}
	\end{align}
	where  $\theta_\ell^r$ ($\theta_\ell^t$) and $\phi_\ell^r$ ($\phi_\ell^t$) denote the azimuth and elevation angles of arrival (departure) associated with the RIS, $\vartheta_\ell^r$ ($\vartheta_\ell^t$) is the angle of arrival (departure) corresponding to the BS (user), and $\alpha_\ell$ ($\beta_\ell$) represents the channel gain, respectively. Furthermore, we assume $L_\mathrm{BR} = L_\mathrm{RU}$ for simplicity. In (\ref{channelModels}), $\mathbf{a}_{\mathrm{BS}}$, $\mathbf{a}_{\mathrm{UE}}$ and $\mathbf{a}_{\mathrm{RIS}}$ denote the $N\times 1$, $\bar{N}\times 1$ and $M\times 1$ normalized steering vectors at the BS, user and RIS, respectively. In particular, for $\mathbf{a}_{\mathrm{BS}}$ and $\mathbf{a}_{\mathrm{UE}}$ with ULA configuration,  the array steering vectors are given by
	\begin{align}
		\mathbf{a}_\mathrm{BS}(\vartheta^r) &= \frac{1}{\sqrt{N}} \left[1, e^{\mathrm{j} \frac{2\pi d}{\lambda}\sin \vartheta^r},\cdots, e^{\mathrm{j} \frac{2\pi d}{\lambda}(N-1)\sin \vartheta^r}\right]^\textsf{T}, \nonumber \\
		\mathbf{a}_\mathrm{UE}(\vartheta^t) &= \frac{1}{\sqrt{\bar{N}}} \left[1, e^{\mathrm{j} \frac{2\pi d}{\lambda}\sin \vartheta^t},\cdots, e^{\mathrm{j} \frac{2\pi d}{\lambda}(\bar{N}-1)\sin \vartheta^t}\right]^\textsf{T},
	\end{align}
	where $d$ and $\lambda$ are the antenna spacing and the signal wavelength. Moreover, for a UPA of size $M = M_y \times M_z$, the  array response of the RIS is given by
	\begin{align}
		\mathbf{a}_\mathrm{RIS}(\theta, &\phi) = \frac{1}{\sqrt{M}} \left[1, \cdots, e^{\mathrm{j}\frac{2\pi d}{\lambda} (m_1-1)\cos \phi \sin \theta + (m_2-1)\sin \phi  }, \right. \nonumber \\
		&\left. \cdots,  e^{\mathrm{j}\frac{2\pi d}{\lambda} (M_y-1)\cos \phi \sin \theta + (M_z-1)\sin \phi  } \right]^\textsf{T},
	\end{align}
	where $m_1 = 1,\cdots, M_y$ and $m_2 = 1,\cdots, M_z$ denote the RIS element index in y- and z-axis, respectively.

	\subsubsection{Channel Acquisition}
	In order to estimate the channels, the transmit pilot signals are used for $t = 1,\cdots, T$ time-slots. Then, the collected received signals at the BS are processed for estimation of the path directions, which are then used to construct the analog beamformers for SPIM. 
	Define $\widetilde{\mathbf{s}}_t$ as  the $\bar{N}\times 1$ uplink transmit pilot signal vector in the $t$-th time-slot, i.e., $\widetilde{\mathbf{s}}_t = \widetilde{\mathbf{A}} \widetilde{\mathbf{B}} \mathbf{s}_t\in \mathbb{C}^{\bar{N}}$, where $ \widetilde{\mathbf{A}}\in \mathbb{C}^{\bar{N}\times N_\mathrm{RF}}$ and $\widetilde{\mathbf{B}}\in \mathbb{C}^{N_\mathrm{RF}\times N_\mathrm{S}}$ denote the analog and baseband beamformers during channel training. Then, the $N\times 1$ received pilot signal at the BS for the $t$-th time-slot is given by
	\begin{align}
		\mathbf{r}_t &= {\mathbf{G}}_\mathrm{BR} \boldsymbol{\Psi}_t {\mathbf{G}}_\mathrm{RU} \widetilde{\mathbf{s}}_t + \mathbf{n}_t, 
		\label{receivedPilot}
	\end{align}
	In order to estimate the path directions at the BS, we exploit the sparsity of the mmWave transmission and represent the uplink channels by using the overcomplete dictionary matrices. Define $\mathbf{D}\in \mathbb{C}^{N\times P_\mathrm{BS}}$, where  $P_\mathrm{BS}$ denotes the number of elements in the dictionary. In particular, we have  $\mathbf{D} = \left[\mathbf{a}_\mathrm{BS}(\vartheta_1^r), \cdots, \mathbf{a}_\mathrm{BS}(\vartheta_{P_\mathrm{BS}}^r)  \right]$, where  $\mathbf{a}_\mathrm{BS}(\vartheta_{p}^r)$ is the $N\times 1$ steering vector corresponding to the direction $\vartheta_{p}^r$ in the dictionary grid for $p = 1,\cdots, P_\mathrm{BS}$. The uplink channel $\ {\mathbf{G}}_\mathrm{BR}$ is then represented in compact form as
	\begin{align}
		\ {\mathbf{G}}_\mathrm{BR} = \mathbf{D}\boldsymbol{\Lambda}_\mathrm{BR} \mathbf{A}_{\mathrm{RIS}}^\textsf{H},
		\label{channelsCompact}
	\end{align}
	where  $\boldsymbol{\Lambda}_\mathrm{BR}$ is a $P_\mathrm{BS}\times L$ $L$-sparse matrix and whose non-zero entries correspond to the path gains  $\alpha_p$. Also, $\mathbf{A}_{\mathrm{RIS}}\in \mathbb{C}^{M\times L}$ is defined as $\mathbf{A}_{\mathrm{RIS}} = \left[\mathbf{a}_\mathrm{RIS}(\theta_1^t,\phi_1^t)\cdots,\mathbf{a}_\mathrm{RIS}(\theta_L^t,\phi_L^t)\right]$. Next, we rewrite the received pilot signal model given in  (\ref{receivedPilot}) in terms of the overcomplete dictionary $\mathbf{D}$ as
	\begin{align}
		\mathbf{r}_t &=     \mathbf{D}\boldsymbol{\Lambda}_\mathrm{BR}\mathbf{A}_{\mathrm{RIS}} ^\textsf{H}\boldsymbol{\Psi}_t {\mathbf{G}}_\mathrm{RU}\widetilde{\mathbf{s}}_t + \mathbf{n}_t.
	\end{align}
	Let us combine  the transmit pilot signals for $t = 1,\cdots, T$  as $\widetilde{\mathbf{S}} = \left[\widetilde{\mathbf{s}}_1,\cdots, \widetilde{\mathbf{s}}_{{T}}\right]\in \mathbb{C}^{\bar{N}\times {T}}$. Then, the collected pilot signals at the BS are given by  $\mathbf{R} = \left[\mathbf{r}_1, \cdots, \mathbf{r}_T\right]\in \mathbb{C}^{N\times T}$, i.e.,
	\begin{align}
		\mathbf{R} &=     \mathbf{D}\boldsymbol{\Lambda}_\mathrm{BR}\mathbf{A}_{\mathrm{RIS}} ^\textsf{H} \widetilde{\mathbf{X}}  + \mathbf{N} \nonumber \\
		& = \mathbf{D}\mathbf{X} + \mathbf{N}, 
		\label{collectedPilots_SU}
	\end{align}
	where $\mathbf{N} = \left[\mathbf{n}_1, \cdots, \mathbf{n}_T\right]\in \mathbb{C}^{N\times T}$, $ \widetilde{\mathbf{X}} = \left[\boldsymbol{\Psi}_1 {\mathbf{G}}_\mathrm{RU}\widetilde{\mathbf{s}}_1,\cdots, \boldsymbol{\Psi}_T {\mathbf{G}}_\mathrm{RU}\widetilde{\mathbf{s}}_T \right]\in \mathbb{C}^{M\times T}$  and  $  {\mathbf{X}} = \boldsymbol{\Lambda}_\mathrm{BR}  \mathbf{A}_{\mathrm{RIS}} ^\textsf{H}\widetilde{\mathbf{X}} \in \mathbb{C}^{P_\mathrm{BS}  \times T }  $. Here, $\mathbf{X}$ is a row-sparse matrix, i.e., only the rows corresponding to the $L$ path directions $\{\vartheta_\ell^r \}_{\ell=1}^L$ are non-zero. Then, $\mathbf{X}$ can be found by solving the following joint-block-sparse recovery problem, i.e.,
	\begin{align}
		&\minimize_{\mathbf{X} } \ \| 	{\mathbf{R}} -\mathbf{D}_\mathrm{BS}\mathbf{X}  \|_\mathcal{F}^2 \nonumber\\
		& \subjectto  \ \| \mathbf{X} \|_{2,0} = L,
		\label{sparseRecovery}
	\end{align}
	where $\| \mathbf{X} \|_{2,0} = \sum_{p = 1}^{P_{\mathrm{BS}}} \mathcal{I}(\|[\mathbf{X}]_{p:}\| \neq 0 ) $ denotes the  $l_{2,0}$-norm  and $\mathcal{I}(x)$ is an indicator function defined as $\mathcal{I}(x) = \left\{\begin{array}{cc}
		1, & x \neq 0 \\
		0, & x = 0
	\end{array}   \right.$~\cite{joint_block_sparsity_Elbir2019Oct}. The problem in (\ref{sparseRecovery}) can be effectively solved via greedy-based techniques, e.g., OMP algorithm. In Algorithm~\ref{alg:OMP}, we present a sparse recovery algorithm by initializing the algorithm with  ${\mathbf{D}} $ and $ {\mathbf{R}}$, which returns the index terms corresponding to the path directions $\hat{\vartheta}_\ell^r$ in $\mathcal{I}$ and the corresponding columns in the dictionary ${\mathbf{D}}$. {The complexity order of the OMP technique in Algorithm \ref{alg:OMP} is mainly due to the matrix multiplications in the Steps 3 ($O\left(2P_\mathrm{BS}NT + P_\mathrm{BS}T^2\right)$) and 5 ($O\left(L^3 + L^2N + 2LNT\right)$), respectively. Thus, the overall complexity order is $O(P_\mathrm{BS}T^2 + 2NT (P_\mathrm{BS}+ L) + L^2(N+1))$.   }

	 Once the sparse recovery problem is solved, the set of steering vectors for SPIM can be constructed as  
	\begin{align}
		{\mathbf{A}}_\mathrm{C} = \left[  \mathbf{a}_\mathrm{BS}(\hat{\vartheta}_1^r), \cdots,  \mathbf{a}_\mathrm{BS}(\hat{\vartheta}_L^r) \right]\in \mathbb{C}^{N\times L}.
		\label{Ac}
	\end{align}
	Note that the sparse recovery algorithm in Algorithm~\ref{alg:OMP} is generic and it can also be used for estimating $\mathbf{A}_\mathrm{UE}$ and $\mathbf{A}_\mathrm{RIS}$ provided that the observation model in (\ref{collectedPilots_SU}) and the optimization problem in (\ref{sparseRecovery}) are constructed accordingly~\cite{irs_CE_Techniques_Wang2020Jun,irs_CE_fast_Zheng2020Nov,irs_CE_and_Loc_Lin2021Sep,irs_channelEst_OMP_Lin2022Apr}.

						\subsubsection{Hybrid Beamformer Design for SPIM}
						
						Once the steering vectors corresponding to the path directions $\{\hat{\vartheta}_\ell^r\}_{\ell=1}^L$ are found, and the set of analog beamformers for SPIM is constructed as in (\ref{Ac}), from which we construct the hybrid beamformer $	\mathbf{F}_i$ as
						\begin{align}
							\mathbf{F}_i = \mathbf{A}_i \mathbf{B}_i = \mathbf{A}_{\mathrm{C}}\mathbf{E}_i  \mathbf{B}_i,
							\label{Fi_Ai_Bi}
						\end{align}
						where  $\mathbf{E}_i$ is an $L\times L_\mathrm{S}$ selection matrix selecting the $L_\mathrm{S}$ out of $L$ steering vectors $\{\mathbf{a}_{\mathrm{BS}}(\hat{\vartheta}_\ell^t)\}_{\ell=1}^L$ for the $i$-th spatial pattern, and $\mathbf{E}_i = \left[\mathbf{e}_{i_1},\cdots, \mathbf{e}_{i_{L_\mathrm{S}}}\right]$,	 where $\mathbf{e}_{i_\ell}$ is the $i_{\ell}$-th column of the identity matrix $\mathbf{I}_{L_\mathrm{S}}$. Then, for the $i$-th spatial pattern, the baseband beamformer $\mathbf{B}_i$ is computed as
						\begin{align}
							\mathbf{B}_i = \mathbf{A}_i^\dagger \mathbf{F}, \label{designB}
						\end{align}
						which is then normalized as $\mathbf{B}_i = \frac{\sqrt{N_\mathrm{S}} {\mathbf{A}_i^\dagger \mathbf{F}  } }{\|\mathbf{A}_i\mathbf{B}_i   \|_\mathcal{F}}$.

						\begin{algorithm}[t]
							\begin{algorithmic}[1] 
								\caption{ \bf Sparse Recovery Algorithm}
								\color{black}
								\Statex {\textbf{Input:} Observation $\mathbf{R}$, dictionary  $\mathbf{D}$, sparsity level $L$. \label{alg:OMP}}
								\State Initialize $\overline{\mathbf{A}}$ as an empty matrix. $\mathcal{I} = \emptyset$. $\mathbf{R}_\mathrm{res} = \mathbf{R}$.
								\State \textbf{while} $k\leq L$ \textbf{do}
								\State \indent $k = \argmax_{\ell} \ \left[ \left(\mathbf{D}^\textsf{H}\mathbf{R}_\mathrm{res} \right) \left(\mathbf{D}^\textsf{H}\mathbf{R}_\mathrm{res} \right) ^\textsf{H}  \right]_{\ell \ell}$.
								\State \indent $\overline{\mathbf{A}} = \left[ \overline{\mathbf{A}}  \  [\mathbf{D}]_k  \right]  $,  $\mathcal{I}  = \mathcal{I}  \cup k $ .
								\State \indent $\mathbf{R}_\mathrm{res} = \mathbf{R} - \overline{\mathbf{A}} \left( \left(\overline{\mathbf{A}}^\textsf{H}\overline{\mathbf{A}}  \right)^\dagger \overline{\mathbf{A}}^\textsf{H} \mathbf{R}_\mathrm{res}  \right) $.
								\State \textbf{end while}
								\Statex \textbf{Output:} $\overline{\mathbf{A}}$ and $\mathcal{I}$.
							\end{algorithmic} 
						\end{algorithm}

						\subsection{RIS Design}
						In order to solve for the RIS phases, we first utilize the FD beamformer $\mathbf{F}$. Then, we construct the equivalent optimization problem maximizing the SE with respect to $\mathbf{F}$ as in (\ref{se_fd}),		for which the optimization problem for $\boldsymbol{\Psi}$ is given by
						\begin{align}
							&\maximize_{\boldsymbol{\Psi}} \; I_\mathrm{FD} \nonumber \\
							&\subjectto \; |[\boldsymbol{\psi}]_m|=1 , \; \forall m ,
							\label{opt_RIS1}
						\end{align}
						which is non-convex due to the constant-modulus constraint on RIS phases. To make the problem more tractable, an upperbound can be written as 
						\begin{align}
							&\log_2\left( \left| \mathbf{I}_{N_\mathrm{S}} +  \frac{1}{\sigma_n^2 N_\mathrm{S}}  \boldsymbol{\Sigma}_1^2 \right|  \right) \overset{{(a)}}{\leq} N_\mathrm{S} \log_2\left( 1 +  \frac{1}{\sigma_n^2 N_\mathrm{S}^2} \mathrm{Tr}\{\boldsymbol{\Sigma}_1^2 \}  \right) \nonumber \\
							& \hspace{67pt} \overset{{(b)}}{\leq} N_\mathrm{S} \log_2\left( 1 +  \frac{1}{\sigma_n^2 N_\mathrm{S}^2} \mathrm{Tr}\{\mathbf{H}\mathbf{H}^\textsf{H} \}  \right),
							\label{upperBound_RIS}
						\end{align}
						where $(a)$ is due to Jensen's inequality and $(b)$ holds since $\mathrm{Tr}\{\mathbf{H}\mathbf{H}^\textsf{H} \} = \sum_{n = 1}^{\mathrm{rank}\{\boldsymbol{\Sigma}\}} \lambda_n^2 \leq \sum_{n = 1}^{N_\mathrm{S}} \lambda_n^2 = \mathrm{Tr}\{\boldsymbol{\Sigma}_1^2 \} $, for which the equality is achieved when $\mathrm{rank}\{\boldsymbol{\Sigma}\} = N_\mathrm{S}$ and $\lambda_n$ denotes the $n$-th diagonal element of $\boldsymbol{\Sigma}$~\cite{ris_tutorial_Wu2021Jan,ris_Design_SU_Zhu}. By utilizing the upperbound in (\ref{upperBound_RIS}) and the cascaded channel in (\ref{channelCascaded}), we get
						\begin{align}
							&\maximize_{\boldsymbol{\psi}} \; \boldsymbol{\psi}^\textsf{H}\mathbf{Q}\boldsymbol{\psi}   \nonumber \\
							&\subjectto \; |[\boldsymbol{\psi}]_m|=1 , \; \forall m ,
							\label{opt_RIS2}
						\end{align}
						where $\mathbf{Q}     $ is an $M\times M$ matrix defined as $\mathbf{Q} = \left(\sum_{n = 1}^{\bar{N}} \mathrm{diag}\{[\mathbf{H}_\mathrm{RU} ]_{n:} \}\mathbf{H}_\mathrm{BR} \mathbf{H}_\mathrm{BR}^\textsf{H} \mathrm{diag}\{[\mathbf{H}_\mathrm{RU} ]_{n:}^\textsf{H} \}\right)^\textsf{T} $, where $[\mathbf{H}_\mathrm{RU} ]_{n:}$ denotes the $n$-th row of $\mathbf{H}_\mathrm{RU} $. 	To effectively solve (\ref{opt_RIS2}), the RIS phase elements can be estimated one by one while the remaining ones are fixed. Define $f_m(\boldsymbol{\psi})$ as the cost function of (\ref{opt_RIS2}) for the $m$-th element of $\boldsymbol{\psi}$, which can be written as 
						\begin{align}
							f_m(\boldsymbol{\psi}) = 	 \sum_{m_1 \neq m }^{M}&\sum_{m_2 \neq m}^{M} [\mathbf{Q}]_{m_1m_2} e^{\mathrm{j}(\psi_{m_1} - \psi_{m_2})} + [\mathbf{Q}]_{mm} \nonumber \\
							& + 		2 \Re\{e^{\mathrm{j}\psi_m}  \sum_{j\neq m}^{M} [\mathbf{Q}]_{mj}e^{-\mathrm{j}\psi_{j}} \} ,
							\label{opt_RIS3}
						\end{align}
						where the first two terms are constant for the fixed $m$-th RIS element. Thus, we can write the problem for the $m$-th element as
						$\maximize_{\psi_m} \Re\{e^{\mathrm{j}\psi_m}  \sum_{j\neq m}^{M} [\mathbf{Q}]_{mj}e^{-\mathrm{j}\psi_{j}} \}  $, which is equivalent to minimizing the phase term of $e^{\mathrm{j}\psi_m} \sum_{j\neq m}^{M} [\mathbf{Q}]_{mj}e^{-\mathrm{j}\psi_{j}}\triangleq \varrho_m e^{-\mathrm{j}(\psi_m - \zeta_m)}$. Then, $\psi_m$ can be computed as  $			\minimize_{\psi_m} \; |\psi_m - \zeta_m|$, for $m = 1,\cdots, M$~\cite{ris_Design_SU_Zhu}. 
						Finally, the discrete phase shifts can be obtained as $\hat{\psi}_m = \left\lfloor \frac{\zeta_m }{ {2\pi }/{2^\Delta } }  \right\rfloor {2\pi }/{2^\Delta }$ for $m = 1,\cdots, M$. {We note here that the selection of the discrete RIS phases is a practical approach as compared to continuous phase selection. As a result, the rate performance is subject to the resolution of the RIS phases.  }

								\subsubsection{Receiver Design}
						\label{sec:receiverDesign}
						At the receiver, the received signal $\mathbf{y}_i$ in (\ref{receivedSignal1}) is processed and the selected index bits are recovered. We assume that the channel path directions are available at the user from the channel acquisition process. Now, suppose that a single RF chain is employed at the user, and define $\mathcal{Z}_i = \{ \{\tilde{\mathbf{a}}_{i,1}, \cdots, \tilde{\mathbf{a}}_{i,L}   \}, \{\tilde{\mathbf{c}}_{i,1},\cdots,\tilde{\mathbf{c}}_{i,L} \} \}$ as the set of analog beamformer pairs for the $i$-spatial pattern, where $\tilde{\mathbf{a}}_{i,\ell}$ and $\tilde{\mathbf{c}}_{i,\ell}$ are the $N\times 1$ and $\bar{N}\times 1$ transmit and receive analog beamformer vectors for the $\ell$-th spatial path, respectively. Next, we define $\mathcal{Z}$ as the set of all beamformer sets associated with all spatial patterns, i.e., $\mathcal{Z} = \{\mathcal{Z}_1,\cdots, \mathcal{Z}_S \}$. The spatial pattern set $\mathcal{Z}$ is available at both transmitter and receiver as a look-out table, from which the recovery of the selected path indices are obtained. The beamformers for the selected $L_\mathrm{S}$ paths, e.g.,  $\{\tilde{\mathbf{a}}_{i,\ell}\}_{\ell=1}^{L_\mathrm{S}}$,  are employed at the transmitter in the $i$-th spatial pattern. Then, the received signal $\mathbf{y}_i$ is used at the receiver to estimate the selected spatial pattern. Then, we compute the strength of the outcome of each receive beamformer $\tilde{\mathbf{c}}_{i,\ell}$ as
						\begin{align}
							\nu_{i,\ell} = \frac{1}{\bar{N}}|\tilde{\mathbf{c}}_{i,\ell}^\textsf{H} \mathbf{y}_i  |, \label{receiver1}
						\end{align}
						which are computed for $\ell = 1,\cdots, L$, and written in a descending order as $\tilde{\nu}_{i,1} > \tilde{\nu}_{i,2} > \cdots > \tilde{\nu}_{i,L}$ without loss of generality, where $\tilde{\nu}_{i,\ell}$ denotes the $\ell$-th element of the ordered set obtained from $\{\nu_{i,\ell}\}_{\ell=1}^L$. Then, the set of recovered index bits are found from the indices of the first $L_\mathrm{S}$ elements of the set $\{\tilde{\nu}_{i,\ell}\}_{\ell=1}^L$. {The algorithmic steps of the overall SPIM scheme is presented in Algorithm~\ref{alg:SPIM}.}

						\begin{algorithm}[t]
							\begin{algorithmic}[1] 
								\caption{ \bf SPIM for RIS-aided mMIMO}
								\State \label{alg:SPIM}Use Algorithm \ref{alg:OMP} and obtain $\overline{\mathbf{A}}$ and $\{ \hat{\vartheta}_{\ell}^r\}_{\ell=1}^{L} $.
								\State Construct analog beamformer $\mathbf{A}_i$ from $\mathbf{A}_\mathrm{C}$ for the $i$-th spatial pattern. 
								\State Design baseband beamformer $\mathbf{B}_i$ as in (\ref{designB}) and construct the hybrid beamformer $\mathbf{F}_i$ as in (\ref{Fi_Ai_Bi}).
								\State Design the RIS phases via (\ref{opt_RIS3}). 
								\State At the user side, detect activated spatial paths and estimate the spatial pattern via \eqref{receiver1}. 
							\end{algorithmic} 
						\end{algorithm}

						\section{Extension to Multi-User Scenario}
						In this section, we extend the SPIM idea to improve the SE performance for multi-user systems. We consider $U$ users with single RF chains for simplicity. Thus, we have $L_\mathrm{S} = 1$ and $N_\mathrm{RF} = U$.  As a result, SPIM is performed by selecting a single path out of $L$ for each user. Also, without loss of generality, we assume that there exist $L$ paths between the RIS and each user. In this case,  the number of paths between the BS and the RIS becomes $\bar{L} = UL$~\cite{irs_CE_Chen2023Feb,irs_CE_MU_Lu2024May}. Note that this does not affect the recovery of the selected spatial path indices as each user perform SPIM recovery separately, which is discussed in \ref{sec:receiverDesign}.

						In the downlink, for a given RIS configuration $\boldsymbol{\Psi}$, the cascaded channel for the $u$-th user is defined similar to  (\ref{channelCascaded}) as $\mathbf{H}_u\in \mathbb{C}^{\bar{N}\times N}$, and the received signal at the $u$-th user becomes
						\begin{align}
							\mathbf{y}_{u,i} = \mathbf{H}_u \mathbf{A}_i \sum_{u=1}^{U}\mathbf{b}_{ui}s_u + \mathbf{n},
						\end{align}
						where  $\mathbf{b}_{ui}\in \mathbb{C}^{N_\mathrm{RF}}$ represents the baseband beamformer vector corresponding to the $u$-th user for the $i$-th spatial pattern. After applying the combiner vector $\mathbf{c}_{ui}\in \mathbb{C}^{\bar{N}}$ ($|[\mathbf{c}_{ui}]_n| = 1/\sqrt{\bar{N}}$, $n =1,\cdots,\bar{N}$), we get the baseband signal ${y}_{u,i} = \mathbf{c}_{ui}^\textsf{H}		\mathbf{y}_{u,i}$ as
						\begin{align}
							{y}_{u,i}  = \mathbf{c}_{ui}^\textsf{H}\mathbf{H}_u \mathbf{A}_i \sum_{u=1}^{U}\mathbf{b}_{ui}s_u + \mathbf{c}_{ui}^\textsf{H}\mathbf{n}.
							\label{sigModel_MU}
						\end{align}
						
						In the following, we introduce the SE expression for multi-user scenario based on the the signal model in (\ref{sigModel_MU}). Then, we present the beamformer and RIS design algorithms for multi-user SPIM.

						\begin{lemma}
							\label{lemma3}
							A closed-form expression for the achievable SE of the multi-user system with SPIM is given by
							\begin{align}
								{I}_\mathrm{SPIM-MU}&\hspace{-2pt} = \hspace{-2pt} \sum_{u = 1}^{U} \hspace{-2pt}\bigg( \hspace{-2pt}\log_2 \frac{S}{2}\hspace{-2pt} -\hspace{-2pt} \frac{1}{S}  \sum_{i = 1}^{S}  \log_2 \hspace{-2pt} \left(\sum_{j = 1}^{S}\hspace{-2pt}\frac{1}{\gamma_{ui} + \gamma_{uj}}  \hspace{-2pt}   \right)\hspace{-4pt} \bigg). \label{SE_7}
							\end{align}
							where $		\gamma_{ui}	= 1 + \frac{ \frac{1}{U}|\mathbf{c}_{u,i}^\textsf{H}\mathbf{H}_u \mathbf{A}_i\mathbf{b}_{ui} s_u  |^2   }{ \frac{1}{U} \sum_{u' \neq u} |\mathbf{c}_{u,i}^\textsf{H}\mathbf{H}_u \mathbf{A}_i\mathbf{b}_{u'i} s_{u'}|^2 +   \sigma_n^2  }$.
						\end{lemma}

						\begin{proof}
							The proof is provided in Appendix~\ref{proofT2}.
						\end{proof}
						
						The beamformer design problem for multi-user scenario can be cast similar to (\ref{opt1}) while maximizing  $I_\mathrm{SPIM-MU}$ in (\ref{SE_7}). Similar to the single-user scenario, we discuss the design of hybrid beamformer and the RIS phases in the following.

						\subsection{Channel Acquisition and Beamforming for Multi-user}
						In order to construct the beamformers, uplink channel acquisition is first performed. Then, the beamformers are constructed similar to the single-user case. In particular, the uplink channels $\widetilde{\mathbf{G}}_\mathrm{BR}\in \mathbb{C}^{N\times M}$ and $\widetilde{\mathbf{G}}_{\mathrm{RU}_u}\in \mathbb{C}^{M\times \bar{N}}$ are defined as
						\begin{align}
							\widetilde{\mathbf{G}}_\mathrm{BR}	& = \sqrt{\frac{N M}{\bar{L}}} \sum_{\ell = 1}^{{L}} \sum_{u = 1}^{U} \alpha_{\ell,u} \mathbf{a}_{\mathrm{BS}}(\vartheta_{\ell,u}^r)\mathbf{a}_{\mathrm{RIS}}^\textsf{H}(\theta_{\ell,u}^t, \phi_{\ell,u}^t) , \nonumber \\
							\widetilde{\mathbf{G}}_{\mathrm{RU}_u}	& = \sqrt{\frac{M\bar{N} }{L}} \sum_{\ell = 1}^{L} \beta_{\ell,u} \mathbf{a}_{\mathrm{RIS}}(\theta_{\ell,u}^r, \phi_{\ell,u}^r)\mathbf{a}_{\mathrm{UE}}^\textsf{H}(\vartheta_{\ell,u}^t) ,
							\label{channelModelsMU}
						\end{align}
						where $\alpha_{\ell,u}\in \mathbb{C}$ and  $\beta_{\ell,u}\in \mathbb{C}$ denote the channel gains, and $\vartheta_{\ell,u}^r,\theta_{\ell,u}^t, \phi_{\ell,u}^t,\theta_{\ell,u}^r, \phi_{\ell,u}^r $ and $\vartheta_{\ell,u}^t$ correspond to the directions of the $\ell$-th path and the $u$-th user.
						
						In order to design the hybrid beamformers at the BS, we need to estimate the path directions $\{\vartheta_{\ell,u}^r \}_{\ell=1,u=1}^{L,U}$, from which the set of analog beamformers for SPIM can be constructed from the corresponding steering vectors $\mathbf{a}_\mathrm{BS}(\vartheta_{\ell,u}^r)$. Thus, we consider the uplink channel acquisition stage, wherein each user transmits the pilot signals $\widetilde{\mathbf{s}}_{t,u}\in \mathbb{C}^{\bar{N}}$ at time-slot $t$. Then, the $N\times 1$ received signal at the BS from the $u$-th user becomes
						\begin{align}
							\widetilde{\mathbf{r}}_{t,u} = \widetilde{\mathbf{G}}_{\mathrm{BR}} \boldsymbol{\Psi}_t \widetilde{\mathbf{G}}_{\mathrm{RU},u}\widetilde{\mathbf{s}}_{t,u} + \mathbf{n}_{t,u}.
						\end{align}
						Then, the received signal collected from all users in the $t$-th time-slot is given by
						\begin{align}
							\widetilde{\mathbf{r}}_{t} & = \sum_{u = 1}^{U}\left(\widetilde{\mathbf{G}}_{\mathrm{BR}} \boldsymbol{\Psi}_t \widetilde{\mathbf{G}}_{\mathrm{RU},u}\widetilde{\mathbf{s}}_{t,u} + \mathbf{n}_{t,u}\right) \nonumber \\
							&= \widetilde{\mathbf{G}}_{\mathrm{BR}} \boldsymbol{\Psi}_t \left(  \sum_{u = 1}^{U} \widetilde{\mathbf{G}}_{\mathrm{RU},u}\widetilde{\mathbf{s}}_{t,u}  \right)  + \sum_{u = 1}^{U} \mathbf{n}_{t,u} \nonumber \\
							& = \widetilde{\mathbf{G}}_{\mathrm{BR}} \boldsymbol{\Psi}_t\overline{\mathbf{G}}_\mathrm{RU} \overline{\mathbf{s}}_t + \widetilde{\mathbf{n}}_t,
						\end{align}
						where $\overline{\mathbf{G}}_\mathrm{RU} = \left[\widetilde{\mathbf{G}}_{\mathrm{RU},1},\cdots, \widetilde{\mathbf{G}}_{\mathrm{RU},U}  \right]\in \mathbb{C}^{M\times \bar{N}U}$, $\overline{\mathbf{s}}_t = \left[\widetilde{\mathbf{s}}_{t,1}^\textsf{T},\cdots, \widetilde{\mathbf{s}}_{t,U}^\textsf{T}   \right]^\textsf{T}\in \mathbb{C}^{\bar{N}U }$ and $\overline{\mathbf{n}}_t =\sum_{u = 1}^{U} \mathbf{n}_{t,u}  \in \mathbb{C}^{N\times 1} $. Furthermore, we can write the received pilot signals for $t= 1, \cdots, T$ time-slots in a compact form as
						\begin{align}
							\overline{\mathbf{R}} = \widetilde{\mathbf{G}}_{\mathrm{BR}} \overline{\mathbf{X}} + \overline{\mathbf{N}},
							\label{pilotsMU}
						\end{align}
						where $\overline{\mathbf{R}} = \left[\overline{\mathbf{r}}_1,\cdots, \overline{\mathbf{r}}_T   \right]\in \mathbb{C}^{N\times T}$, $ \overline{\mathbf{X}}  = \left[\boldsymbol{\Psi}_1 \overline{\mathbf{G}}_\mathrm{RU}\overline{\mathbf{s}}_1, \cdots,  \boldsymbol{\Psi}_T \overline{\mathbf{G}}_\mathrm{RU}\overline{\mathbf{s}}_T\right]\in \mathbb{C}^{M\times T}$  and $\overline{\mathbf{N}} = \left[\overline{\mathbf{n}}_1,\cdots, \overline{\mathbf{n}}_T   \right]\in \mathbb{C}^{N\times T}$.  The signal model in (\ref{pilotsMU}) is similar to the collected pilot signals for the single-user case in (\ref{collectedPilots_SU}), for which the channels $\widetilde{\mathbf{G}}_{\mathrm{BR}}$, $ \widetilde{\mathbf{G}}_{\mathrm{RU},u}$ and the path directions can be estimated via CS-based techniques~\cite{irs_CE_Techniques_Wang2020Jun,irs_channelEst_OMP_Lin2022Apr,irs_CE_practical_MU_Di2020Feb} or tensor-based techniques~\cite{irs_MU_MIMO_CE_deAraujo2021Feb}. {Specifically, one can employ an OMP-based approach presented in~\cite{elbir_low_res_hB_10742011} which takes into account the inter-user interference in multi-user scenario.         }

						 Define $\hat{\vartheta}_{\ell,u}^r$ as the estimated path directions for $\ell = 1,\cdots, L$ and $u = 1,\cdots, U$. Then,  the set of analog beamformer at the BS for SPIM is given by
						\begin{align}
							\widetilde{\mathbf{A}}_\mathrm{C} = \left[\mathbf{a}_\mathrm{BS}(\hat{\vartheta}_{1,1}^r), \cdots, \mathbf{a}_\mathrm{BS}(\hat{\vartheta}_{L,1}^r), \cdots, \mathbf{a}_\mathrm{BS}(\hat{\vartheta}_{L,U}^r)  \right],
						\end{align}
						from which we construct the hybrid beamformer as
						\begin{align}
							\widetilde{\mathbf{F}}_i = \widetilde{\mathbf{A}}_i \widetilde{\mathbf{B}}_i = \widetilde{\mathbf{A}}_\mathrm{C} \widetilde{\mathbf{E}}_i \widetilde{\mathbf{B}}_i,
						\end{align}
						where $\widetilde{\mathbf{E}}_i = \left[\mathbf{E}_{1,i},\cdots, \mathbf{E}_{U,i}   \right] $ is a $UL\times UL_\mathrm{S}$ selection matrix, where
						the $L\times L_\mathrm{S}$ matrix $\mathbf{E}_{u,i}$ performs the selection of $L_\mathrm{S}$ out of $L$ paths for the  $u$-th user, and it is defined similar to $\mathbf{E}_i$.
						
						{In multi-user scenario, the detection of the spatial pattern at the users is similar to the single-user case as discussed in Section~\ref{sec:receiverDesign}. Specifically, each user has the knowledge of spatial path directions which are obtained during channel estimation stage. Then, a single RF chain can be activated at the user for the detection of the activated spatial path directions. Define $\tilde{\mathbf{c}}_{u,i,\ell}\in \mathbb{C}^{\bar{N}}$ as the analog combiner corresponding to the $\ell$-th path in the $i$-th spatial pattern of the $u$-th user. Then, the strength of the outcome of each receive combiner is computed as $	\nu_{u,i,\ell} = \frac{1}{\bar{N}} |\tilde{\mathbf{c}}_{u,i,\ell}\mathbf{y}_{u,i}  |  $ for $\ell = 1,\cdots, L$. By utilizing the ordered set  $ \tilde{\nu}_{u,i,1} > \tilde{\nu}_{u,i,2} > \cdots > \tilde{\nu}_{u,i,L}  $, where $\tilde{\nu}_{u,i,\ell}$  is the $\ell$-th element in the set $\{\nu_{u,i,\ell} \}_{\ell=1}^L$ in descending order, we get the recovered index bits at the user $u$ from the indices of the first $L_\mathrm{S}$ elements.  
				    }

						\subsection{RIS Design for Multi-user}
						In multi-user scenario, we optimize the RIS phases by taking into account that the SINR constraint as $\mathrm{SINR}_{ui}\geq \kappa_u$, where $\kappa_u$ is defined as the minimum SINR threshold. Then, the RIS phase design problem is given by
						\begin{align}
							\mathrm{Find} \ & \  \boldsymbol{\psi} \nonumber \\
							\subjectto & \ \mathrm{SINR}_{ui} \geq \kappa_u \nonumber \\
							&\subjectto \; |[\boldsymbol{\psi}]_m|=1 , \; \forall m ,
							\label{opt_RIS_MU}
						\end{align}
						Define $\mathbf{z}_{u_1,u_2,i} = \mathrm{diag}\{\mathbf{c}_{u_1,i}^{H}\mathbf{H}_{\mathrm{RU},u_2} \}\mathbf{H}_\mathrm{BR}\mathbf{A}_i\mathbf{b}_{u_2,i} \in \mathbb{C}^{M} $, for which we have $\mathbf{c}_{u,i}^\textsf{H}\mathbf{H}_u \mathbf{A}_i\mathbf{b}_{ui} = \boldsymbol{\psi}^\textsf{H} \mathbf{z}_{u_1,u_2,i}  $. Then,  $\mathrm{SINR}_{ui}$ is rewritten as $		\mathrm{SINR}_{u,i} = \frac{\frac{1}{U} |\boldsymbol{\psi}^\textsf{H}\mathbf{z}_{u,u,i}   |^2         }{ \frac{1}{U}   \sum_{u' \neq u} |\boldsymbol{\psi}^\textsf{H}\mathbf{z}_{u,u',i}|^2 + \sigma_n^2             }$. Then, we can write the following equivalent problem to (\ref{opt_RIS_MU}) as
						\begin{align}
							\mathrm{Find} \ & \  \boldsymbol{\psi} \nonumber \\
							\subjectto & {\boldsymbol{\psi}}^\textsf{H} \mathbf{Z}_{u,u,i}{\boldsymbol{\psi}} \geq  \kappa_u \sum_{u' \neq u} \left( {\boldsymbol{\psi}}^\textsf{H} \mathbf{Z}_{u,u',i}{\boldsymbol{\psi}}  + \sigma_n^2\right)   \nonumber \\
							& \; |[\boldsymbol{\psi}]_m|=1 , \; \forall m ,
							\label{opt_RIS_MU2}
						\end{align}
						where $\mathbf{Z}_{u_1,u_2,i} = 
						\mathbf{z}_{u_1,u_2,i} \mathbf{z}_{u_1,u_2,i}^\textsf{H}    \in \mathbb{C}^{M\times M}  $~\cite{irs_TWC}. Next, we aim to rewrite (\ref{opt_RIS_MU2}) in a more tractable form by leveraging semi-definite relaxation (SDR) techniques. Therefore, we rewrite the inequality constraint in (\ref{opt_RIS_MU2}) by using the equality that $\mathrm{Tr}\{{\boldsymbol{\psi}}^\textsf{H} \mathbf{Z}_{u,u,i}{\boldsymbol{\psi}}\} = \mathrm{Tr}\{\mathbf{Z}_{u,u,i}{\boldsymbol{\psi}} {\boldsymbol{\psi}}^\textsf{H} \}$. Thus, we define $\boldsymbol{\Theta} = \boldsymbol{\psi} \boldsymbol{\psi}^\textsf{H} \in \mathbb{C}^{M\times M}$ satisfying $\boldsymbol{\Theta}\succeq
						\mathbf{0} $ and $\mathrm{rank}\{\boldsymbol{\Theta} \} = 1$. Then, the problem in (\ref{opt_RIS_MU2}) can be recast as
						\begin{align}
							\mathrm{Find} \ & \  \boldsymbol{\Theta} \nonumber \\
							\subjectto &\mathrm{Tr}\{\mathbf{Z}_{u,u,i}\boldsymbol{\Theta} \} \geq  \kappa_u \sum_{u' \neq u} \left(\mathrm{Tr}\{\mathbf{Z}_{u,u',i}\boldsymbol{\Theta} \}  + \sigma_n^2\right)   \nonumber \\
							& \ [\boldsymbol{\Theta}]_{mm}   = 1 , \; \forall m,
							\ \boldsymbol{\Theta}\succeq
							\mathbf{0},
							\mathrm{rank}\{\boldsymbol{\Theta}\} = 1,
							\label{opt_RIS_MU3}
						\end{align}
						which is still non-convex due to the rank constraint, which can be removed, and the problem in (\ref{opt_RIS_MU3}) can be effectively solved via Gaussian randomization techniques and a feasible solution can be found~\cite{irs_TWC,sdr_Luo2010Apr}.

						\section{Numerical Simulations}
						In this section, we present the performance of the proposed SPIM approach via various numerical experiments in terms of SE over $500$ Monte Carlo trials. The number of antennas at the BS and the user are $N=128$ and $\bar{N}=16$, respectively. The channel path directions are drawn from $[-90^\circ,90^\circ]$ uniformly at random, while the path gains are drawn from $ \mathcal{N}(1, (0.2)^2)$.

						\begin{figure}[t]
							\centering
							{\includegraphics[draft=false,width=.99\columnwidth]{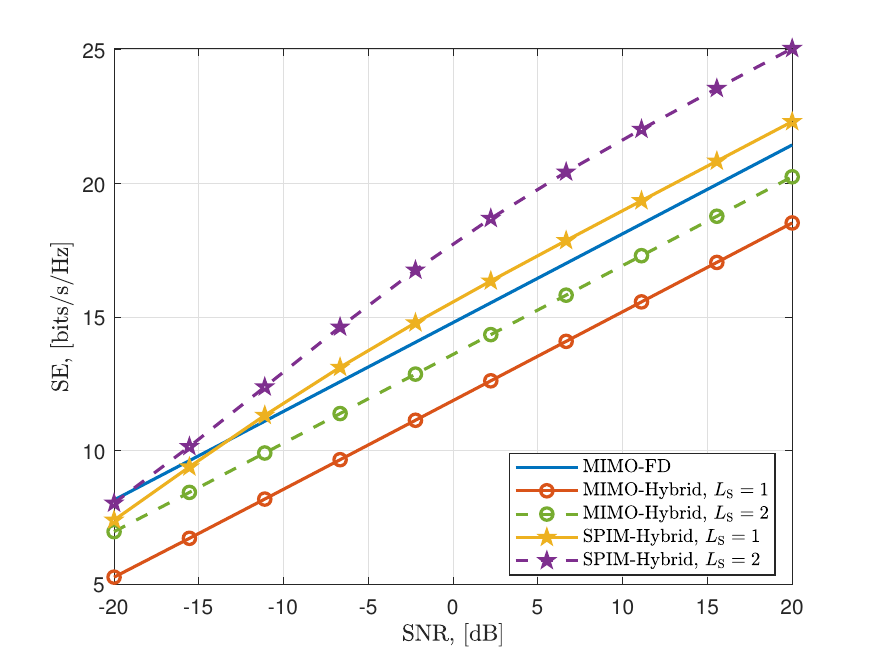} }
							\caption{SE versus $\mathrm{SNR}$  when  $L=8$  for $L_\mathrm{S} = \{1,2\}$.
							}
							\label{fig_SE_SNR}
						\end{figure}

						\begin{figure}[t]
							\centering
							{\includegraphics[draft=false,width=.99\columnwidth]{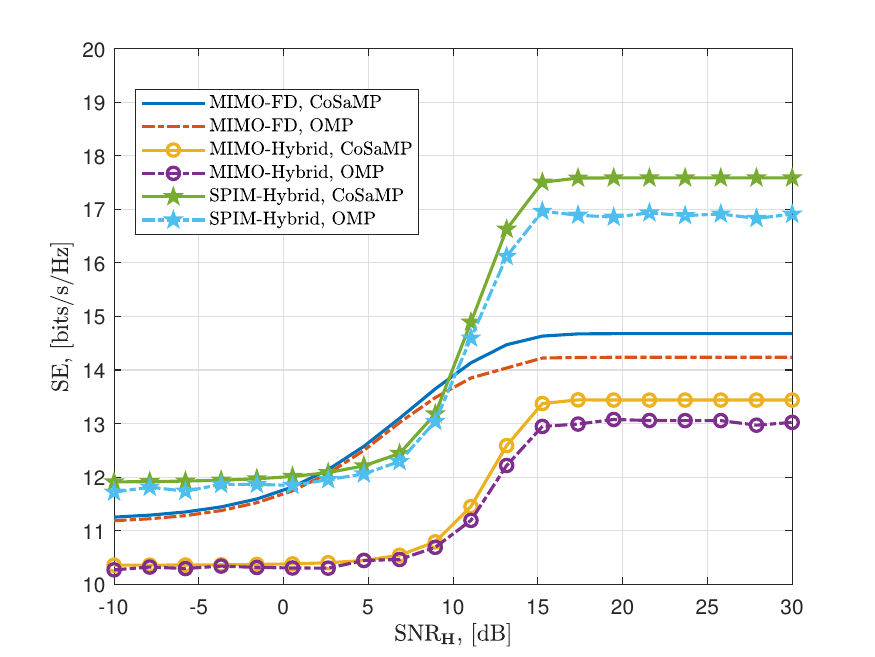} }
							\caption{SE versus $\mathrm{SNR}_\mathbf{H}$  when  $L=8$  for $L_\mathrm{S} = 1$.
							}
							\label{fig_SE_SNR_H}
						\end{figure}

						\begin{figure}[t]
							\centering
							{\includegraphics[draft=false,width=.99\columnwidth]{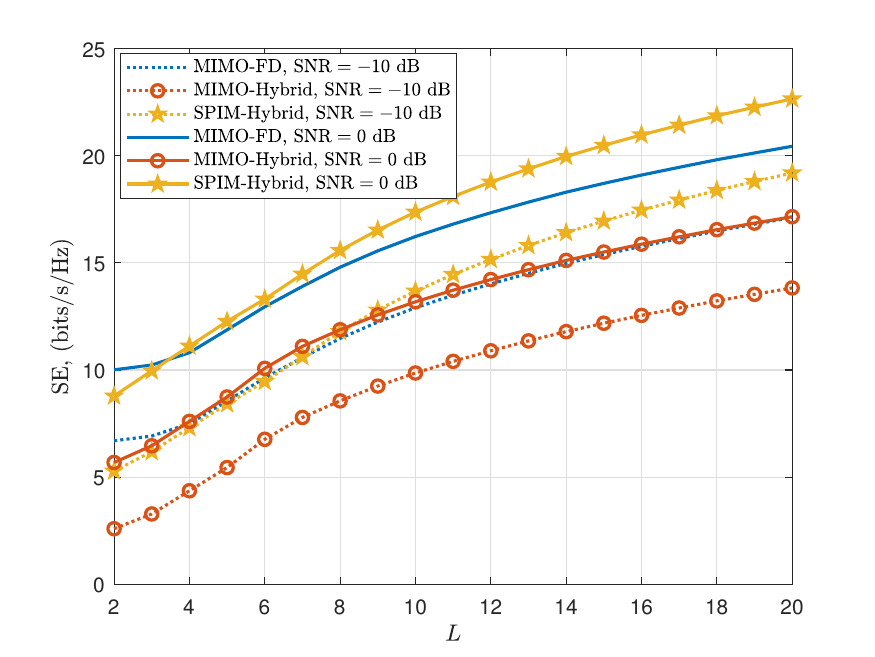} }
							\caption{SE versus number of paths $L$ for $\mathrm{SNR}=\{-10,0\}$ dB and $L_\mathrm{S}= N_\mathrm{S} = 1$.
							}
							\label{fig_SE_L}
						\end{figure}
						
						Fig.~\ref{fig_SE_SNR} shows the SE performance of our SPIM approach in comparison with conventional hybrid and FD beamformers when $L=8$ and $L_\mathrm{S}\in \{1,2\}$. Conventional hybrid beamforming achieves less SE due to the use of limited RF chains as the analog beamformers are designed corresponding to the strongest 1 and 2 channel paths $L_\mathrm{S}=1$ and $L_\mathrm{S}=2$, respectively. While our SPIM approach employs the same number of RF chains, it exploits the spatial diversity of $L=8$ paths, and delivers additional index bits corresponding to multiple spatial patterns. This allows a significant increase in the SE, even higher than that of FD beamformer. In particular, $20\%$ SE improvement is obtained via SPIM than that of FD beamforming when $\mathrm{SNR}=0$.

						{Fig.~\ref{fig_SE_SNR_H} shows the SE performance of the proposed SPIM scheme under the impact of imperfect channels. Here, the imperfect channel is obtained by introducing AWGN onto the perfect channel matrix, and define $\mathrm{SNR}_\mathbf{H}$ as the SNR on the channel matrices given in (\ref{channelCascaded}). We also present the performance of different sparse recovery algorithms, e.g., OMP and CoSaMP~\cite{omp_cosamp_CE_9238899}.  We can see from Fig~\ref{fig_SE_SNR_H} that the performance of all beamforming schemes improves as $\mathrm{SNR}_\mathbf{H}$ increases while the proposed SPIM approach outperforms the FD beamformer for $\mathrm{SNR}_\mathbf{H}>10$ dB. We also observe that CoSaMP-based beamforming performs better than the OMP algorithm~\cite{omp_cosamp_CE_9238899,cosamp_needell2009cosamp}. The superior performance of CoSaMP is thanks to its support selection strategy. Specifically, in OMP, which can select only a single support vector in each iteration, the selected supports are not removed from the dictionary matrix whereas CoSaMP overcomes this defect by selecting more support vectors and removes them for the next iteration. 
						We note that the support recovery algorithm presented in Algorithm~\ref{alg:OMP} is based on OMP, and it provides a generic solution for the support recovery problem. Thus, any sparse support recovery algorithm can be employed in our proposed SPIM approach.  }

						Fig.~\ref{fig_SE_L} illustrates the performance of our SPIM approach with respect to $L$ when $\mathrm{SNR}=0$ dB and $L_\mathrm{S}= N_\mathrm{S} = 1$. We observe that the SPIM-based hybrid beamforming achieves higher SE than the FD-based beamforming for $L>3$. This means that $S=3$ different spatial pattern is enough to achieve higher than FD beamforming performance while using only a single RF chain. 
						
						We also present the comparison of the algorithms if different number of spatial paths are selected for SPIM. In Fig.~\ref{fig_SE_Ls}, we present the SE of the algorithms with respect to $L_\mathrm{S}$. We can see that the performance of conventional hybrid beamformer approaches to the FD beamformer as it employs a dedicated RF chain for each path, i.e., $N_\mathrm{RF}=L_\mathrm{S}$. Likewise, the SE of our SPIM approach also increases. However, it reduces and attains the SE of FD beamformer since the number of possible spatial patterns $S$ reduces as $L_\mathrm{S}\rightarrow L$. {Apart from the values of $L $ and $L_\mathrm{S}$, the channel gain of the spatial paths also affects the SE performance. Specifically, SPIM is favorable when the path gains are comparable~\cite{spim_bounds_JSTSP_Wang2019May,elbir_spim_radarConf_10149629}. In order to evaluate this scenario, we consider the scenario for $L=2$ and $L_\mathrm{S}=1$, and assume that the channel paths gains have the relationship as $\alpha_1+ \alpha_2 =1 $. Fig.~\ref{fig_SE_gamma} shows the SE performance with respect to $\alpha_1$ of SPIM and conventional MIMO hybrid beamforming. It can be seen that SPIM achieves higher SE than that of MIMO hybrid beamforming as $\alpha_1 <4 \alpha_2$~\cite{spim_bounds_JSTSP_Wang2019May}. When the spatial path gains are apart from each other, i.e., $\alpha_1 >4\alpha_2$, the conventional MIMO beamforming becomes favorable. This is expected as MIMO beamforming employs the strongest spatial paths.  }

						\begin{figure}[t]
							\centering
							{\includegraphics[draft=false,width=.99\columnwidth]{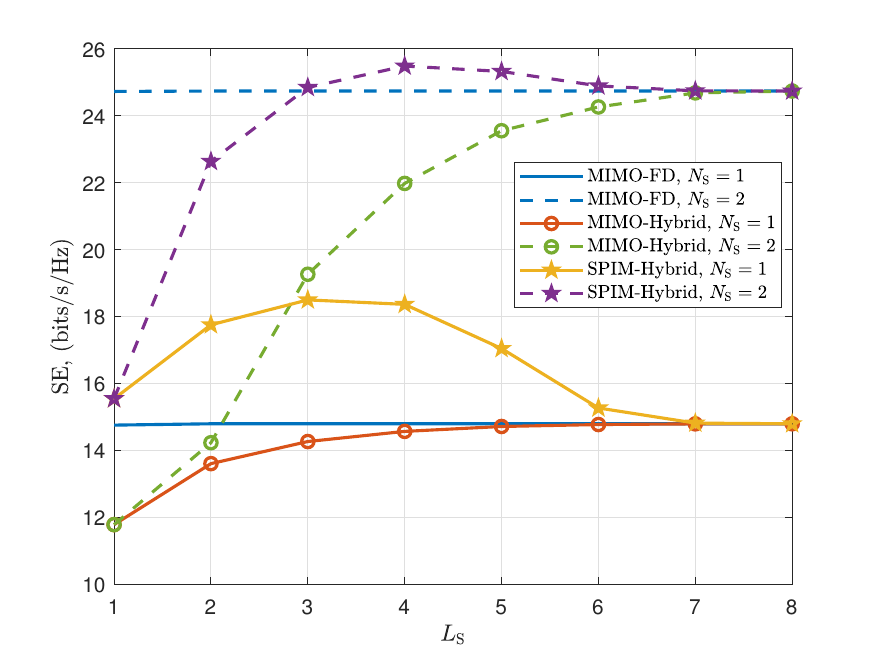} }
							\caption{SE versus $L_\mathrm{S}$ when $\mathrm{SNR}=0$ dB, $L =8$ and $ N_\mathrm{S} = \{1,2\}$.
							}
							\label{fig_SE_Ls}
						\end{figure}
						
						\begin{figure}[t]
							\centering
							{\includegraphics[draft=false,width=.99\columnwidth]{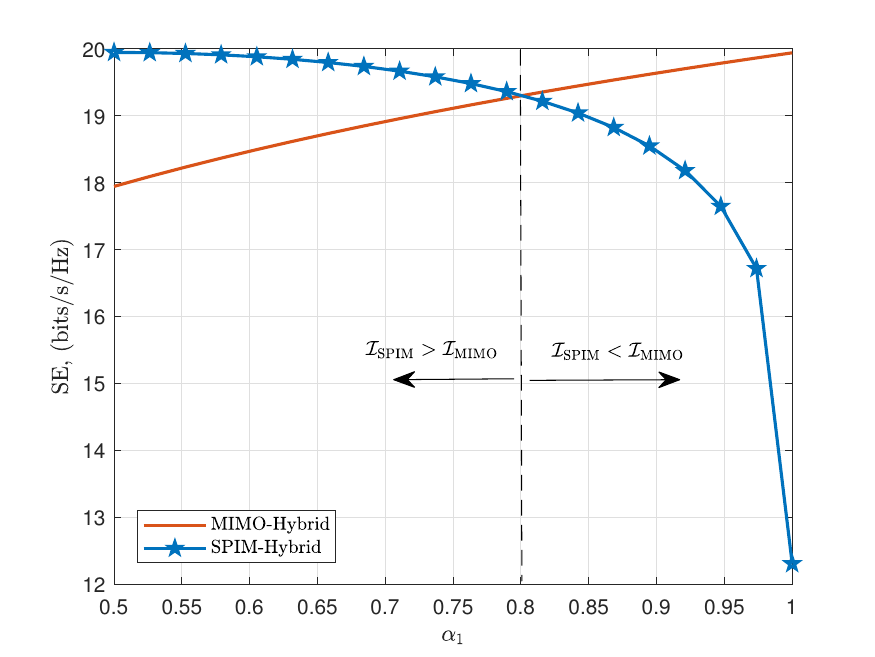} }
							\caption{SE versus spatial path gain $\alpha_1$ when $L=2$ and $L_\mathrm{S}=1$.
							}
							\label{fig_SE_gamma}
						\end{figure}

						In Fig. \ref{fig_SE_M}, the SE performance is presented with respect to the number of RIS elements $M$ for $L\in \{2,4,8\}$ when the number of selected paths $L_\mathrm{S}=1$. 
						%

						{Fig~\ref{fig_SE_bound_L} shows the performance gap between SPIM and FD beamforming techniques in comparison with the performance bound given in the right hand side of (\ref{thereom1eq}). We see that the lower bound is tighter  for the lower values of $L$ while it remains constant and  SE performance difference, i.e.,  $\mathcal{I}_\mathrm{SPIM}-\mathcal{I}_\mathrm{FD}$ becomes large and increases as the total number of spatial paths $L$ increases.   }

						\begin{figure}[t]
							\centering
							{\includegraphics[draft=false,width=.99\columnwidth]{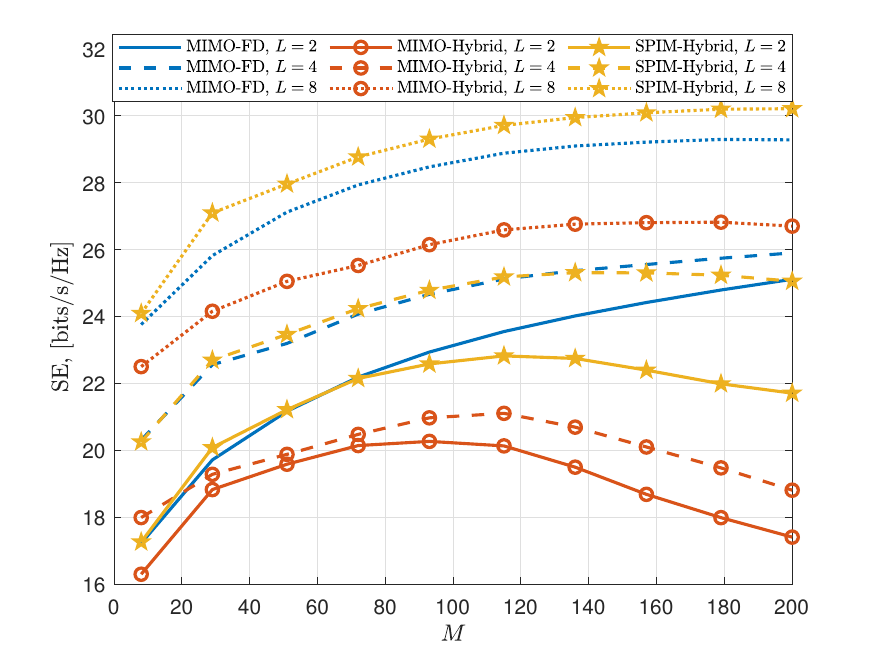} }
							\caption{SE versus $M$ when $\mathrm{SNR}=0$ dB, $L =\{2,4,8\}$ and $ L_\mathrm{S}=1$.
							}
							\label{fig_SE_M}
						\end{figure}

						\begin{figure}[t]
							\centering
							{\includegraphics[draft=false,width=.99\columnwidth]{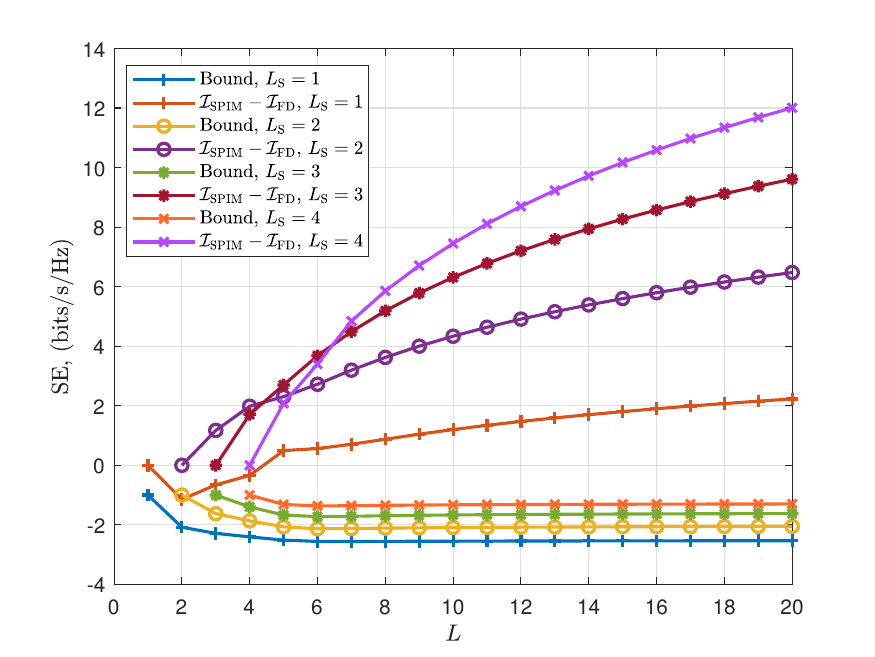} }
							\caption{SE performance bound and $\mathcal{I}_\mathrm{SPIM}-\mathcal{I}_\mathrm{FD}$ for $L_\mathrm{S}=\{1,2,3,4\}$   when $\mathrm{SNR}=0$ dB.
							}
							\label{fig_SE_bound_L}
						\end{figure}

						\begin{figure}[t]
							\centering
							{\includegraphics[draft=false,width=.99\columnwidth]{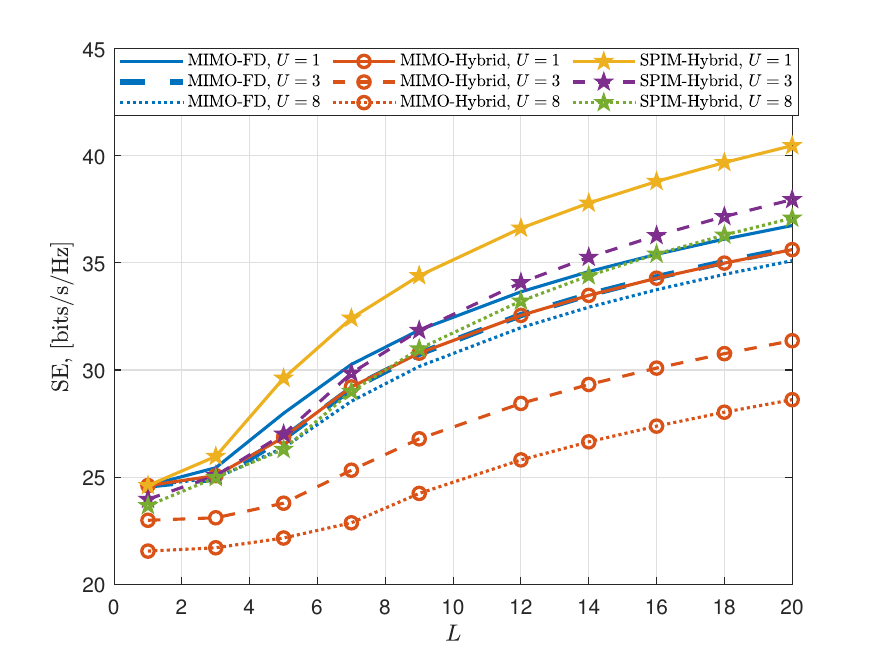} }
							\caption{SE analysis in multi-user scenario with respect to  $L$ for $U =\{1,3,8\}$ when $\mathrm{SNR}=0$ dB.
							}
							\label{fig_SE_MU_L}
						\end{figure}
						
						\begin{figure}[t]
							\centering
							{\includegraphics[draft=false,width=.99\columnwidth]{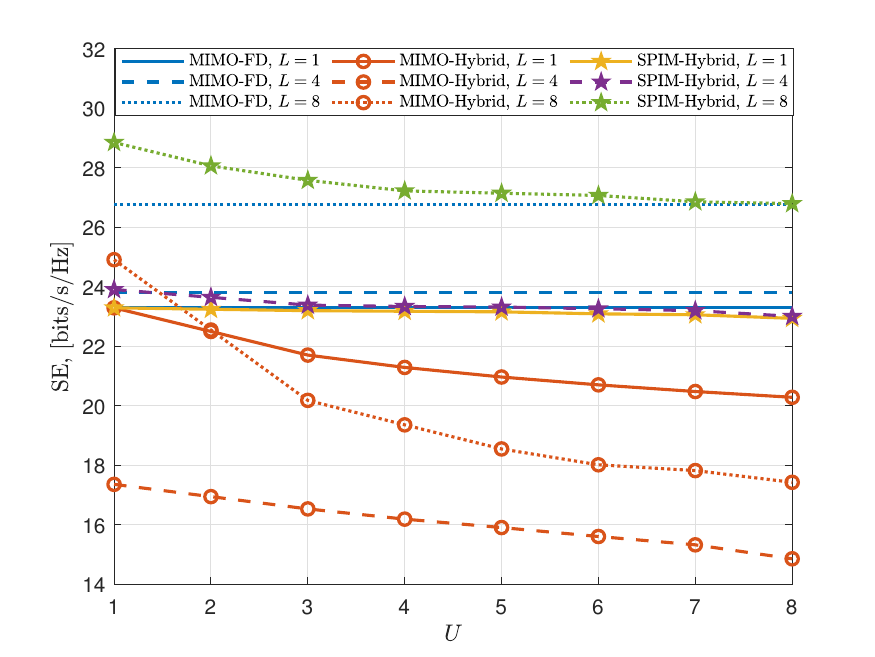} } 
							\caption{SE analysis in multi-user scenario with respect to  $U$ for $L = \{1,4,8\}$ when $\mathrm{SNR}=0$ dB.
							}
							\label{fig_SE_MU_U}
						\end{figure}

						Next, we perform numerical simulations for  multi-user scenario. In this experiment, each user employs a single RF chain while the number of RF chains at the BS equals to the number of users, i.e., $N_\mathrm{RF}=U$. Therefore, we have $L_\mathrm{S}=1$ spatial paths selected from $L$ paths for each user. Fig.~\ref{fig_SE_MU_L} shows the SE performance with respect to the $L$ for $U\in \{1,3,8\}$. We can see the improvement in the SE as $L$ increases in the multi-user scenario while a loss in the SE occurs as the number of users increases because of the interference among the users, which deteriorates the performance of digital beamformer that accounts for the inter-user interference (IUI). Nevertheless, the proposed SPIM-based hybrid beamforming achieves higher SE than that of FD multi-user beamforming while a single RF chain per user is employed. Fig.~\ref{fig_SE_MU_U} presents the SE analysis with respect to the number of users for $L\in \{1,4,8\}$. The SE reduces as the number of users increases because of IUI for hybrid beamforming techniques with/without SPIM while it is fixed for the FD beamformer which is not affected by the IUI.

						\section{Conclusions}
						We introduced an SPIM framework for RIS-aided mmWave systems, wherein the hybrid beamformers are designed by exploiting the spatial scattering paths between the BS and the user through the RIS. We provided theoretical analysis on the achievable SE relationship between the SPIM-based hybrid beamforming and FD beamforming. Then, we validated our theoretical analysis and showed that our SPIM approach can achieve higher SE than that of FD beamforming while employing a few RF chains. The performance of the proposed approach is attributed to exploiting the spatial diversity of the wireless channel between the BS and the users in RIS-aided mMIMO systems. We have shown that the SPIM-aided mMIMO is favorable in the multipath environments involving more than 4 spatial paths, which is usually the case in mmWave mMIMO scenario. The numerical experiments also revealed that better than FD beamforming performance can be achieved by using only one of two RF chains when there are only 8 spatial paths in the environment.


						


						\appendices

						\section{Proof of Lemma \ref{lemma1}} 
						\label{proofL1}
						Rewrite $	I_\mathrm{SPIM}$ in (\ref{SE_14}) for $S = 1$ as
						\begin{align}
							I_\mathrm{SPIM} &= \log_2\left(\frac{1}{(2\sigma_n^2)^{\bar{N}}}\right)   -   \log_2  \left(\frac{1}{ | 2\mathbf{M}_1 |  }  \right)\nonumber\\
							&= -\bar{N} \log_2(2\sigma_n^2)    -   \log_2  \left(\frac{1}{ 2^{\bar{N}}| \mathbf{M}_1 |  }  \right)\nonumber\\
							& =  -\bar{N}  -\bar{N} \log_2(\sigma_n^2)  +   \log_2 (2^{\bar{N}}| \mathbf{M}_1 |)   \nonumber\\
							& =  -\bar{N}  -\bar{N} \log_2(\sigma_n^2) +  \log_2 \left((2\sigma_n^2)^{\bar{N}}\left| \frac{\mathbf{M}_1}{\sigma_n^2}\right|\right)  \nonumber\\
							& =  -\bar{N}  -\bar{N} \log_2(\sigma_n^2) +  \log_2 \left((2\sigma_n^2)^{\bar{N}}\right) + \log_2 \left(\left| \frac{\mathbf{M}_1}{\sigma_n^2}\right|\right)  \nonumber\\
							& =   \log_2 \left(\frac{1}{\sigma_n^2}\left| \mathbf{M}_1\right|\right),   \label{SE_15}
						\end{align}
						from which we obtain (\ref{se_mmWave}). \qed

						\section{Proof of Theorem \ref{thereom1}}
						\label{proofT1}
						
						We consider the hybrid beamformers $\mathbf{F}_i = \mathbf{A}_i\mathbf{B}_i$ and the FD beamformer $\mathbf{F}$ for the SPIM-aided and the FD MIMO systems, respectively. As the hybrid beamformers are optimized to maximize the SE of the communication system, we can assume that the mmWave MIMO system parameters for the hybrid beamformers, i.e.,  $\mathbf{F}_i$ and the FD beamformer $\mathbf{F}$ are ``close", i.e.,  the eigenvalues of the matrix $\mathbf{I}_{N_\mathrm{S}} - \mathbf{V}_1^\textsf{H}\mathbf{F}_i\mathbf{F}_i^\textsf{H}  $ are small for $i \in \mathcal{S}$~\cite{mimoRHeath}. Then, we can write the the SE of the mmWave MIMO beamforming for a given hybrid beamformer $\mathbf{A}_i\mathbf{B}_i$ as
						\begin{align}
							&	I_\mathrm{MIMO} = \log_2 \left(\frac{1}{\sigma_n^2}\left| \mathbf{M}_i  \right|  \right) \nonumber \\
							& =\underbrace{\log_2\left( \left| \mathbf{I}_{N_\mathrm{S}} 
								+ \frac{1}{\sigma_n^2 N_\mathrm{S}} \boldsymbol{\Sigma}_1^2  \right|  \right) }_{I_\mathrm{FD}}- (N_\mathrm{S} - || \mathbf{V}_1^\textsf{H}\mathbf{F}_i ||_\mathcal{F}^2), \label{se_mmwave}
						\end{align}
						from which it is clear that $I_\mathrm{MIMO} = 	I_\mathrm{FD} $ when $\mathbf{F}_i = \mathbf{F}$ and $|| \mathbf{V}_1^\textsf{H}\mathbf{F} ||_\mathcal{F}^2 = N_\mathrm{S}$. Using (\ref{se_mmwave}), we get
						\begin{align}
							|\mathbf{M}_i| = \sigma_n^2 \left| \mathbf{I}_{N_\mathrm{S}} 
							+ \frac{1}{\sigma_n^2 N_\mathrm{S}} \boldsymbol{\Sigma}_1^2  \right| 2^{-( N_\mathrm{S} - u_i) },\label{mi_2}
						\end{align}	
						where $u_i = || \mathbf{V}_1^\textsf{H}\mathbf{F}_i ||_\mathcal{F}^2$. Consider the SE of SPIM in (\ref{SE_7}), which can be rewritten as
						\begin{align}
							&I_\mathrm{SPIM} = \log_2\left(\frac{S}{2}\right)   -  \frac{1}{S}  \sum_{i = 1}^{S}  \log_2  \left(\sum_{j = 1}^{S}\frac{1}{ | \boldsymbol{\Xi}_{ij} |  }  \right), \label{spim_1}
						\end{align}
						where the determinant of $\boldsymbol{\Xi}_{ij}\in \mathbb{C}^{N_\mathrm{R}\times N_\mathrm{R}}$ is 
						\begin{align}
							|	\boldsymbol{\Xi}_{ij}|& =  \left |		2\mathbf{I}_{N_\mathrm{R}} + \frac{1}{\sigma_n^2 N_\mathrm{S}} \mathbf{H}\mathbf{A}_i\mathbf{B}_i \mathbf{B}_i^\textsf{H}\mathbf{A}_i^\textsf{H}\mathbf{H}^\textsf{H} \right. \nonumber \\
							&\left. \hspace{40pt} +  \frac{1}{\sigma_n^2 N_\mathrm{S}} \mathbf{H}\mathbf{A}_j\mathbf{B}_j \mathbf{B}_j^\textsf{H}\mathbf{A}_j^\textsf{H}\mathbf{H}^\textsf{H} \right|\nonumber \\
							&=\left|	2\mathbf{I}_{N_\mathrm{R}} + \frac{1}{\sigma_n^2 N_\mathrm{S}} \mathbf{H} \left[\boldsymbol{\Pi}_i + \boldsymbol{\Pi}_j \right] \mathbf{H}^\textsf{H}\right| \nonumber\\ 
							& = 2\left|	\mathbf{I}_{N_\mathrm{R}} + \frac{1}{2\sigma_n^2 N_\mathrm{S}} \mathbf{H} \left[\boldsymbol{\Pi}_i + \boldsymbol{\Pi}_j \right] \mathbf{H}^\textsf{H}\right| \label{xi1}
						\end{align}
						where $\boldsymbol{\Pi}_i = \mathbf{A}_i\mathbf{B}_i \mathbf{B}_i^\textsf{H}\mathbf{A}_i^\textsf{H} \in \mathbb{C}^{N_\mathrm{T}\times N_\mathrm{T}} $. Using (\ref{mi_2}), we can write $\log_2\left(\left| 	\boldsymbol{\Xi}_{ij}  \right|  \right)$ as
						\begin{align}
							&\log_2\left(\left| 	\boldsymbol{\Xi}_{ij}  \right|  \right) =  \log_2\left( \left| \mathbf{I}_{N_\mathrm{S}} 
							+ \frac{1}{\sigma_n^2 N_\mathrm{S}} \boldsymbol{\Sigma}_1^2  \right|  \right) \nonumber \\
							& - \mathrm{tr}\{\mathbf{I}_{N_\mathrm{S}} - \mathbf{V}_1^\textsf{H}[\boldsymbol{\Pi}_i + \boldsymbol{\Pi}_j ] \mathbf{V}_1   \}, \label{xi3}
						\end{align}
						where  $\mathrm{tr}\{\mathbf{I}_{N_\mathrm{S}} - \mathbf{V}_1^\textsf{H}[\boldsymbol{\Pi}_i + \boldsymbol{\Pi}_j ] \mathbf{V}_1   \} = N_\mathrm{S} - \mathrm{tr}\{\mathbf{V}_1^\textsf{H}\boldsymbol{\Pi}_i\mathbf{V}_1\} - \mathrm{tr}\{\mathbf{V}_1^\textsf{H}\boldsymbol{\Pi}_j\mathbf{V}_1\}   \} = N_\mathrm{S} - || \mathbf{V}_1^\textsf{H}\mathbf{A}_i \mathbf{B}_i ||_\mathcal{F}^2 - || \mathbf{V}_1^\textsf{H}\mathbf{A}_j \mathbf{B}_j ||_\mathcal{F}^2 $. Thus, (\ref{xi3}) becomes 
						\begin{align}
							&\log_2\left(\left| 	\boldsymbol{\Xi}_{ij}  \right|  \right) =  \log_2\left( \left| \mathbf{I}_{N_\mathrm{S}} 
							+ \frac{1}{\sigma_n^2 N_\mathrm{S}} \boldsymbol{\Sigma}_1^2  \right|  \right) \nonumber \\
							& - (N_\mathrm{S} - u_i - u_j), \label{xi4}
						\end{align}
						where $u_z = ||\mathbf{V}_1^\textsf{H}\mathbf{A}_z \mathbf{B}_z ||_\mathcal{F}^2 $ for $z \in \{i,j\}$.		Then, using (\ref{xi4}), (\ref{xi1}) is written as
						\begin{align}
							|	\boldsymbol{\Xi}_{ij}| = 2 \left| \mathbf{I}_{N_\mathrm{S}} 
							+ \frac{1}{2\sigma_n^2 N_\mathrm{S}} \boldsymbol{\Sigma}_1^2  \right| 2^{-( N_\mathrm{S} - u_i-u_j) }. \label{xi2}
						\end{align}
						Substituting (\ref{xi2}) into (\ref{spim_1}) yields
						\begin{align}
							&I_\mathrm{SPIM} =   \log_2\left(\frac{S}{2}\right)    \nonumber \\
							&-  \frac{1}{S}  \sum_{i = 1}^{S}  \log_2  \left(\sum_{j = 1}^{S}\frac{2^{( N_\mathrm{S} - u_i-u_j)}}{  2 \left| \mathbf{I}_{N_\mathrm{S}} 
								+ \frac{1}{2\sigma_n^2 N_\mathrm{S}} \boldsymbol{\Sigma}_1^2  \right| }  \right) \nonumber\\
							& = \log_2\left(\frac{S}{2}\right)  \nonumber\\
							& -\frac{1}{S}\sum_{i = 1}^{S} \log_2 \left( \frac{1}{  2 \left| \mathbf{I}_{N_\mathrm{S}} 
								+ \frac{1}{2\sigma_n^2 N_\mathrm{S}} \boldsymbol{\Sigma}_1^2  \right| } \sum_{j = 1}^S 2^{( N_\mathrm{S} - u_i-u_j)}  \right) \nonumber  \\
							& = \log_2\left(\frac{S}{2}\right)  +\frac{1}{S}\sum_{i = 1}^{S} \log_2\left( 2 \left| \mathbf{I}_{N_\mathrm{S}} 
							+ \frac{1}{2\sigma_n^2 N_\mathrm{S}} \boldsymbol{\Sigma}_1^2  \right|  \right) \nonumber
						\end{align}
						\begin{align}
							& -\frac{1}{S}\sum_{i = 1}^{S} \log_2\left(\sum_{j = 1}^S 2^{( N_\mathrm{S} - u_i-u_j)}  \right) \nonumber 
						\end{align}
						\begin{align}
							& = \log_2\left(\frac{S}{2}\right) + \log_2 \left(2 \left|\mathbf{I}_{N_\mathrm{S}} 
							+ \frac{1}{2\sigma_n^2 N_\mathrm{S}} \boldsymbol{\Sigma}_1^2\right|  \right) \nonumber 
						\end{align}
						\begin{align}
							&-\frac{1}{S}\sum_{i = 1}^{S} \log_2\left(\sum_{j = 1}^S 2^{( N_\mathrm{S} - u_i-u_j)}  \right) \nonumber 
						\end{align}
						\begin{align}
							& = \log_2\left(\frac{S}{2}\right) + \log_2 \left(\frac{1}{2}\left| 2\mathbf{I}_{N_\mathrm{S}} 
							+ \frac{1}{\sigma_n^2 N_\mathrm{S}} \boldsymbol{\Sigma}_1^2  \right|\right) \nonumber
						\end{align}
						\begin{align}
							&-\frac{1}{S}\sum_{i = 1}^{S} \log_2\left(\sum_{j = 1}^S 2^{( N_\mathrm{S} - u_i-u_j)}  \right) \nonumber \\
							& = \log_2\left(\frac{S}{4}\right) + \log_2 \left(\left| 2\mathbf{I}_{N_\mathrm{S}} 
							+ \frac{1}{\sigma_n^2 N_\mathrm{S}} \boldsymbol{\Sigma}_1^2  \right|\right)  \nonumber 
						\end{align}
						\begin{align}
							&-\frac{1}{S}\log_2\left(\prod_{i = 1}^{S}\sum_{j = 1}^S 2^{( N_\mathrm{S} - u_i-u_j)}  \right) \nonumber \\
							& = \log_2\left(\frac{S}{4}\right) + \log_2 \left(\left| 2\mathbf{I}_{N_\mathrm{S}} 
							+ \frac{1}{\sigma_n^2 N_\mathrm{S}} \boldsymbol{\Sigma}_1^2  \right|\right)  \nonumber \\
							&- N_\mathrm{S} -\frac{1}{S}\log_2\left(\prod_{i = 1}^{S}\sum_{j = 1}^S 2^{-(   u_i+u_j)}  \right), \label{spim2}
						\end{align}
						wherein we can write the following inequality for the second term in (\ref{spim2}), i.e.,
						\begin{align}
							\left| 2\mathbf{I}_{N_\mathrm{S}} 
							+ \frac{1}{\sigma_n^2 N_\mathrm{S}} \boldsymbol{\Sigma}_1^2  \right| \geq \left| \mathbf{I}_{N_\mathrm{S}} \right| 
							+  \left|\mathbf{I}_{N_\mathrm{S}} +  \frac{1}{\sigma_n^2 N_\mathrm{S}} \boldsymbol{\Sigma}_1^2  \right|,
						\end{align}
						from which we can get
						\begin{align}
							&\log_2 \left(\left| 2\mathbf{I}_{N_\mathrm{S}} 
							+ \frac{1}{\sigma_n^2 N_\mathrm{S}} \boldsymbol{\Sigma}_1^2  \right| \right) 
							\geq 	 	I_\mathrm{FD}. \label{inequality_fd}
						\end{align}
						Substituting (\ref{inequality_fd}) into (\ref{spim2}) yields the expression in (\ref{thereom1eq}), which completes the proof. \qed

						\section{Proof of Lemma~\ref{lemma3}}
						\label{proofT2}
							The achievable rate of the overall multi-user systems is 
							\begin{align}
								\mathcal{R} = \sum_{u = 1}^{U} I({y}_{ui};\mathbf{s}, \mathbf{b}_{ui}), \label{SE_9}
							\end{align}
							where $\mathbf{s} = \left[s_1,\cdots, s_U\right]^\textsf{T}$ and  $I({y}_{ui};\mathbf{s}|\mathbf{b}_{ui})$ represents the MI between the $i$-th spatial pattern and the received signal ${y}_{ui}$ at the $u$-th user as in (\ref{sigModel_MU}), and it can be expressed as
							\begin{align}
								I({y}_{ui};\mathbf{s}, \mathbf{b}_{ui}) = I({y}_{ui};\mathbf{s}|\mathbf{b}_{ui}) + I({y}_{ui};\mathbf{b}_{ui}), \label{SE_k}
							\end{align}
							where $I({y}_{ui};\mathbf{s}|\mathbf{b}_{ui})$ and $I({y}_{ui};\mathbf{b}_{ui})$ correspond to the information bits conveyed via conventional APM and SPIM, respectively. In particular, $I({y}_{ui};\mathbf{s}|\mathbf{b}_{ui})$ corresponds to the conventional symbol transmission, and it can be quantified by using Shannon's CCMC capacity~\cite{capacity_Telatar1999Nov} as
							\begin{align}
								I({y}_{ui};\mathbf{x}|\mathbf{b}_{ui}) = \frac{1}{S}\sum_{i = 1}^{S}\log_2 \left(  {\gamma_{ui}} \right), \label{SE_8}
							\end{align}
							where $\gamma_{ui} =  1 + \mathrm{SINR}_{ui} $, for which the $\mathrm{SINR}_{ui}$ is given by
							\begin{align}
								\mathrm{SINR}_{ui} = \frac{ \frac{1}{U}|\mathbf{c}_{u,i}^\textsf{H}\mathbf{H}_u \mathbf{A}_i\mathbf{b}_{ui} s_u  |^2   }{ \frac{1}{U} \sum_{u' \neq u} |\mathbf{c}_{u,i}^\textsf{H}\mathbf{H}_u \mathbf{A}_i\mathbf{b}_{u'i} s_{u'}|^2 +   \sigma_n^2  }.
							\end{align}	In (\ref{SE_k}), the spatial domain term $I({y}_{ui};\mathbf{b}_{ui})$ corresponds to the information bits conveyed via SPIM, and it is written as~\cite{informationTheory_book_Cover2005Apr}
							\begin{align}
								I({y}_{ui};\mathbf{b}_{ui}) =& \frac{1}{S} \sum_{i = 1}^{S} \bigg(
								\int f({y}_{ui} |i) \nonumber \\
								& \times  \log_2 \left( \frac{f({y}_{ui}|i)}{\frac{1}{S}\sum_{j = 1}^{S}f({y}_{ui}|j)  }  \right) d{y}_{ui}  \bigg), \label{SE_20}
							\end{align}
							where $f({y}_{ui}|i) = \mathcal{CN}({0},\gamma_{ui})$ is the conditional probability distribution of ${y}_{ui}$ given that the $i$-th spatial pattern is selected~\cite{spim_SE_He2017May}.
							Next, we rewrite (\ref{SE_20}) as 
							\begin{align}
								I({y}_{ui}&;\mathbf{b}_{ui}) = \frac{1}{S} \sum_{i = 1}^{S} \bigg( \int f({y}_{ui}|i) \log_2 f({y}_{ui}|i) d{y}_{ui}   \nonumber \\ 
								& -  \int f({y}_{ui} |i) \log_2 \left( \frac{1}{S}\sum_{j = 1}^{S}f({y}_{ui}|j)    \right) d{y}_{ui}  \bigg). \label{SE_2}
							\end{align} 
							By computing the entropy of the random variable ${y}_{ui}$, we get the first term in the right hand side of (\ref{SE_2}) as~\cite{capacity_Telatar1999Nov} 
							\begin{align}
								\frac{1}{S}\sum_{i = 1}^S\int f({y}_{ui} |i) &\log_2 f({y}_{ui}|i) d{y}_{ui}  \nonumber\\
								& = \frac{1}{S}\sum_{i = 1}^S (- \log_2 (\pi e ) - \log_2(\gamma_{ui}) ). \label{SE_3}
							\end{align}
							The last  term in (\ref{SE_2}), i.e., $\int f({y}_{ui}|i)\log_2 \left( \frac{1}{S}\sum_{j=1}^{S}f({y}_{ui}|j) \right)d{y}_{ui} $, cannot be expressed by a closed-form formulation. However, it can be upperbounded by exploiting the concavity of $\log_2(\cdot)$ as
							\begin{align}
								&\int f({y}_{ui}|i)\log_2 \left( \frac{1}{S}\sum_{j=1}^{S}f({y}_{ui}|j) \right)d{y}_{ui} \nonumber \\
								& \overset{\mathrm{(a)}}{\leq} \log_2  \left(\frac{1}{S}\sum_{j=1}^{S}  f({y}_{ui} |i) f({y}_{ui}|j) d{y}_{ui} \right) \nonumber \\
								& = - \log_2 \pi  + \log_2 \left(\frac{1}{S}\sum_{j = 1}^{S}\frac{1}{\gamma_{ui} + \gamma_{uj}}     \right), \label{SE_4}
							\end{align}
							where $(a)$ is obtained by applying  Jensen's inequality~\cite{spim_onGSM_He2017Sep}. Combining (\ref{SE_3}) and (\ref{SE_4}), we get a lower bound for $	I({y}_{ui};\mathbf{b}_{ui})$  as
							\begin{align}
								I({y}_{ui};\mathbf{b}_{ui}) \geq  \log_2 \left(\frac{S}{e}\right)- \frac{1}{S}  \sum_{i = 1}^{S}  \log_2  \left(\sum_{j = 1}^{S}\frac{\gamma_{ui}}{\gamma_{ui} + \gamma_{uj}}     \right). \label{SE_5}
							\end{align}
							The asymptotic performance of the expression in (\ref{SE_5}) is biased. Specifically, as ${\rho/\sigma_n^2 \rightarrow 0}$ and ${\rho/\sigma_n^2 \rightarrow \inf}$, we have $	I({y}_{ui};\mathbf{b}_{ui}) {\rightarrow} 0  $ and  $I({y}_{ui};\mathbf{b}_{ui}) {\rightarrow} \log_2 S  $ while the right hand side of (\ref{SE_5}) equals to $1 -\log_2 e$ and $\log_2 S + 1 -  \log_2 e$, respectively~\cite{spim_SE_He2017May,spim_bounds_JSTSP_Wang2019May}. In order to obtain an unbiased approximation of $	I({y}_{ui};\mathbf{b}_{ui})$, the asymptotic bias, i.e., $1 -\log_2 e$ is compensated, and we get
							\begin{align}
								I({y}_{ui};\mathbf{b}_{ui}) \approx  \log_2 \left(\frac{S}{2}\right)- \frac{1}{S}  \sum_{i = 1}^{S}  \log_2  \left(\sum_{j = 1}^{S}\frac{\gamma_{ui}}{\gamma_{ui} + \gamma_{uj}}     \right). \label{SE_6}
							\end{align}
							By substituting $	I({y}_{ui};\mathbf{b}_{ui})$ and $I({y}_{ui};\mathbf{s}|\mathbf{b}_{ui})$ into (\ref{SE_k}) and (\ref{SE_9}), we obtain the closed form approximation in (\ref{SE_7}), which completes the proof.   \qed
							

						
						\footnotesize
						\bibliographystyle{IEEEtran}
						\bibliography{references_141}

							\begin{IEEEbiography}[\frame{\includegraphics[width=1.0in,height=1.25in,clip]{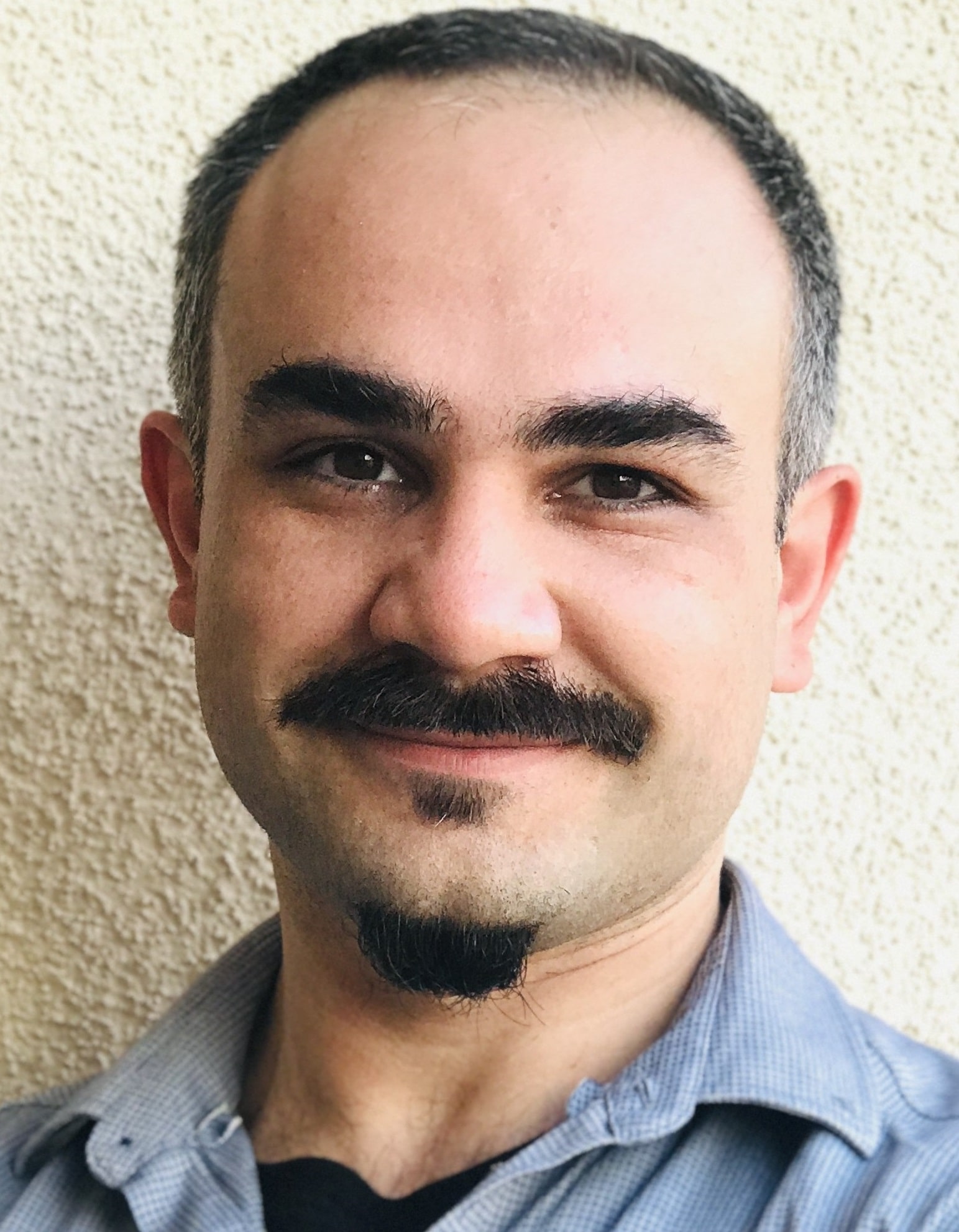}}]{Ahmet M. Elbir} (Senior Member, IEEE) received the B.S. degree with Honors from Firat University in 2009, and the Ph.D. degree from the Middle East Technical University (METU) in 2016, both in electrical engineering. He is currently an Associate Professor at Istinye University. He is also a Research Fellow at University of Luxembourg since 2021. In the past, he was a Research Fellow at King Abdullah University of Science and Technology (2023-2025), the University of Hertfordshire (2021-2023), a Postdoctoral Researcher at Koc University  (2020-2021) and Carleton University (2022-2023). He also worked as a Senior Researcher at Duzce University (2016-2022). His research interests include array signal processing for radar systems and wireless communications, and deep learning for multi-antenna systems. He currently serves as an Associate Editor for \textsc{IEEE Transactions on Signal Processing} and \textsc{IEEE Transactions on Aerospace and Electronic Systems} and a Senior Area Editor for \textsc{IEEE Signal Processing Letters}. He served as an Associate Editor for \textsc{IEEE Wireless Communications Letters}, and a Lead Guest Editor for the special issues in \textsc{IEEE Journal of Selected Topics in Signal Processing}, \textsc{IEEE Signal Processing Magazine}, \textsc{IEEE Wireless Communications}, \textsc{IEEE Communications and Standards Magazine}, \textsc{IEEE Transactions on Antennas and Propagation}, \textsc{IEEE Journal of Selected Topics in Applied Earth Observations and Remote Sensing} and Elsevier Signal Processing. He is also the Chair of IEEE SPS Synthetic Aperture Technical Working Group (TWG), a member of IEEE SPS Sensor Array and Multichannel (SAM) Technical Committee, member of Steering Committee for \textsc{IEEE Wireless Communications Letters}  and a member of IEEE AESS Integrated Sensing and Communications TWG. Dr. Elbir is the recipient of the METU Best Ph.D. thesis award (2016), the IEEE Turkey Section Research Encouragement Award (2022), the IET Radar, Sonar \& Navigation Best Paper Award (2022) and the IEEE AESS Harry Rowe Mimno Award (2026).
						\end{IEEEbiography}

						\begin{IEEEbiography}[\frame{\includegraphics[width=1.0in,height=1.25in,clip]{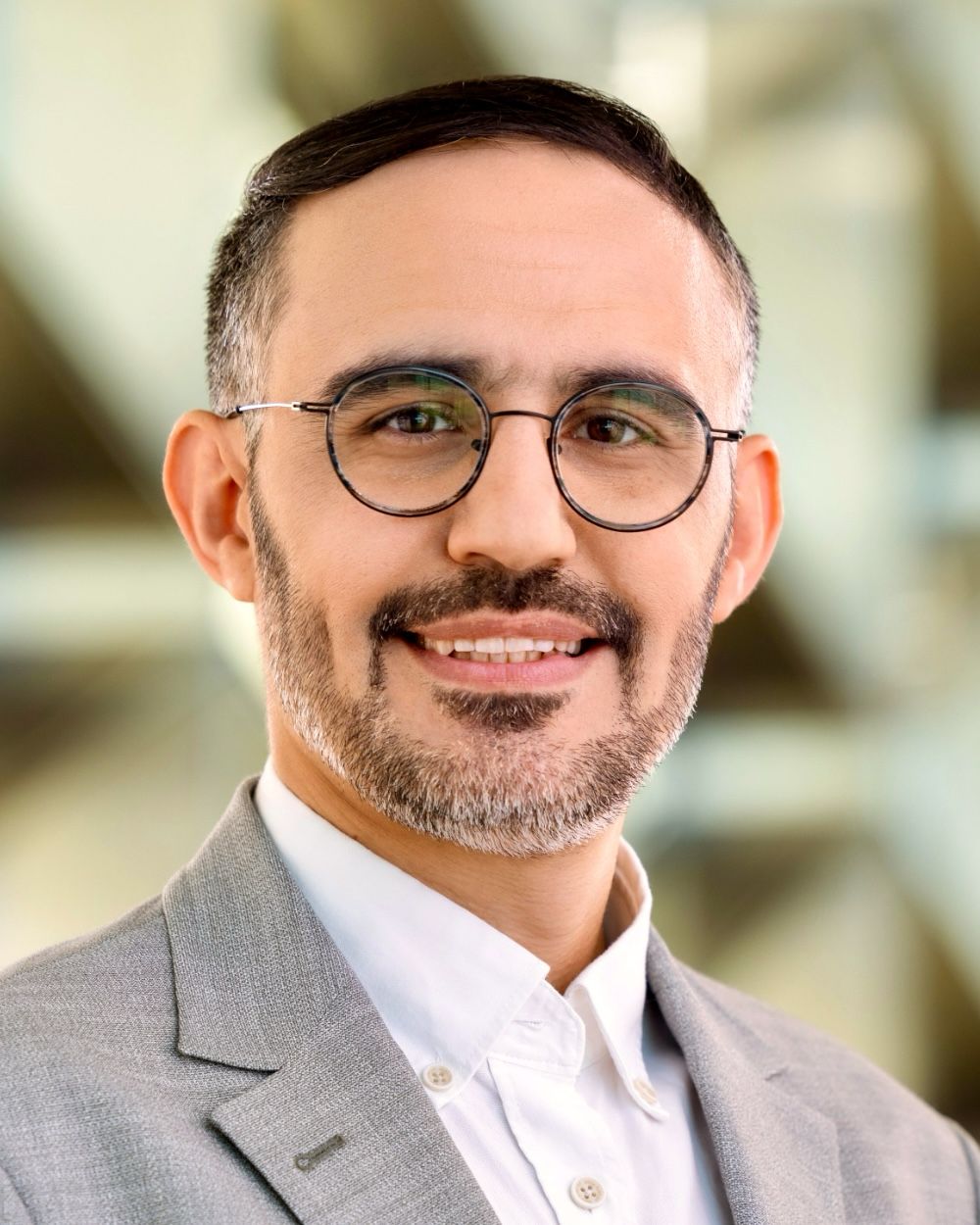}}]{Abdulkadir Celik} (Senior Member, IEEE) 
							is an Associate Professor in the School of Electronics and Computer Science at the University of Southampton, UK, where he also serves as the Director of the Centre for Internet of Things and Pervasive Systems. He received the Ph.D. degree in double-majors of Electrical Engineering and Computer Engineering from Iowa State University, Ames, IA, USA, in 2016; wherein he also earned M.S. degrees in Electrical Engineering and Computer Engineering in 2013 and 2015, respectively. Prior to his current appointment, he was a senior research scientist from 2020 to 2025 and a post-doctoral fellow from 2016 to 2020 at King Abdullah University of Science and Technology (KAUST), Thuwal, KSA. Dr. Celik is the recipient of \textsc{IEEE} Communications Society’s 2023 Outstanding Young Researcher Award for Europe, Middle East, and Africa (EMEA) region. He currently serves as an editor for npj Wireless Technology, \textsc{IEEE Transactions on Communications}, \textsc{IEEE Communications Letters}, \textsc{IEEE Wireless Communication Letters}, and Frontiers in Communications and Networks. His research interests are in the broad areas of next-generation wireless communication systems and networks.
						\end{IEEEbiography}
						
							\begin{IEEEbiography}[\frame{\includegraphics[width=1.0in,height=1.25in,clip]{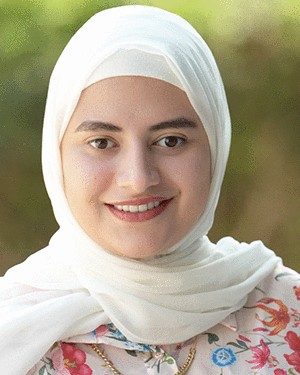}}]{Asmaa Abdallah} (Member, IEEE) received the BS (with High Distinction) and MS degrees in computer and communications engineering from Rafik Hariri University (RHU), Lebanon, in 2013 and 2015, respectively, and the PhD degree in electrical and computer engineering ffrom the American University of Beirut (AUB), Beirut, Lebanon. From 2021 to 2024, she was a postdoctoral fellow with the King Abdullah University of Science and Technology (KAUST), where she is currently a research scientist with the Communications and Computing Systems Laboratory. Her research interests include machine learning, communication theory, stochastic geometry, array signal processing, with emphasis on energy and spectral efficient algorithms for next-generation wireless communication systems. Dr. Abdallah was the recipient of the Academic Excellence Award at RHU in 2013 for ranking first on the graduating class, and the scholarship from the Lebanese National Counsel for Scientific Research (CNRS-L/AUB) to support her doctoral studies. In 2023, Dr. Abdallah has been selected by MIT technology review as one of the leading 15 Innovators under 35 in the MENA area. From 2016 to 2020, she was a member of the executive committee of IEEE Young Professionals Lebanon’s Section
						\end{IEEEbiography}

							\begin{IEEEbiography}[\frame{\includegraphics[width=1.0in,height=1.25in,clip]{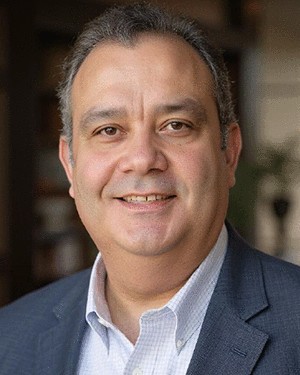}}] {Ahmed M. Eltawil} (Senior Member, IEEE) received the bachelor’s and master’s degrees from Cairo University, Giza, Egypt, in 1997 and 1999, respectively, and the Doctorate degree from the University of California, Los Angeles, in 2003. From 2005 to 2021, he was a professor of electrical engineering and computer science with University of California, Irvine (UCI), USA. He is currently a professor and associate dean for Research with Computer, Electrical, and Mathematical Sciences and Engineering (CEMSE) Division, King Abdullah University of Science and Technology (KAUST), where he established the Communication and Computing Systems Laboratory (CCSL) to conduct research on efficient architectures for computing and communications systems, with a particular focus on mobile wireless systems. His research interests include various application domains, such as low-power mobile systems, machine learning platforms, sensor networks, body area networks, and critical infrastructure networks. He was the recipient of the several recognitions and awards, including the US National Science Foundation CAREER Award, 2021 ”Innovator of the Year” award by the Henry Samueli School of Engineering at the University of California, Irvine, and two United States Congress certificates of merit, among other recognitions. He was also on numerous editorial roles over the years, and an expert reviewer for national and international funding agencies and review boards. He was a distinguished lecturer for IEEE COMSOC during the 2023/24 term. He also holds senior membership in the National Academy of Inventors in the United States.
						\end{IEEEbiography}

					\end{document}